\documentclass[envcountsect,runningheads,orivec]{llncs}
\usepackage[T1]{fontenc}
\usepackage[utf8]{inputenc}
\usepackage{amsmath,amsfonts,amssymb}
\usepackage{stmaryrd}
\SetSymbolFont{stmry}{bold}{U}{stmry}{m}{n}
\usepackage{scalerel}
\usepackage{proof}
\usepackage{mathtools}
\usepackage{needspace}
\usepackage{hyperref}
\usepackage{color}

\hypersetup{
	colorlinks=true,
	linkcolor=blue,
	filecolor=magenta,
	urlcolor=cyan,
}

\allowdisplaybreaks

\newcommand\CPM{\ensuremath{\mathsf{CPM}}}
\newcommand\D[1][n]{\ensuremath{\mathcal D_{#1}}}
\newcommand\X{\ensuremath{\mathsf{X}}}
\newcommand\Z{\ensuremath{\mathsf{Z}}}
\providecommand\bra[1]{\ensuremath{\langle #1|}}
\providecommand\braket[2]{\ensuremath{\langle #1|#2 \rangle}}
\newcommand\cnot{\ensuremath{\mathsf{Cnot}}}
\newcommand\ext[1]{\ensuremath{\overline{#1}}}
\newcommand\fsem[2][\theta]{\ensuremath{\llbracket #2\rrbracket_{#1}}}
\newcommand\had{\ensuremath{\mathsf{H}}}
\newcommand\hbigotimes{\ensuremath{\widehat{\bigotimes}}}
\newcommand\hotimes{\ensuremath{\mathbin{\hat\otimes}}}
\newcommand\hsum{\ensuremath{\widehat{\sum}}}

\providecommand\ket[1]{\ensuremath{|{#1}\rangle}}
\newcommand\ketbra[2]{\ensuremath{|{#1}\rangle\!\langle{#2}|}}
\newcommand\letxk[4]{\ensuremath{\mathsf{let}\ {#1}^{\otimes #2}  = #3\ \mathsf{in}\ #4}}
\newcommand\letxn[3]{\ensuremath{\mathsf{let}\ {#1}^{\otimes n}  = #2\ \mathsf{in}\ #3}}
\newcommand\qlambdens{\ensuremath{\lambda_\rho^\circ}}
\newcommand\qletcase[3]{\ensuremath{\mathsf{letcase}^\circ\ #1 = #2\ \mathsf{in}\ \{#3\}}}

\newcommand\tr{\ensuremath{\mathsf{tr}}}
\newcommand\tsem[1]{\ensuremath{\llbracket #1\rrbracket}}

\newcommand{\SN}{\mathsf{SN}}

\begin{document}

\title{A quantum let within the lambda calculus}
\titlerunning{A quantum let within the lambda calculus}
\author{Alejandro Díaz-Caro\inst{1,2}
  \and
  Tomás Miguez \inst{3,}\thanks{This paper is based on the second's author Master's thesis~\cite{Miguez26}.}
}
\authorrunning{A.~Díaz-Caro and T.~Miguez}

\institute{Université de Lorraine, CNRS, Inria, LORIA, F-54000 Nancy, France \and
  InCo, FIng, Universidad de la República, Colonia del Sacramento, Uruguay \\
  \email{adiazcaro@fing.edu.uy}\\
\and
DC, FCEN, Universidad de Buenos Aires, CABA, Argentina\\
\email{tomasmiguez99@gmail.com}}
\maketitle              
\begin{abstract}
Since the seminal work of Selinger and Valiron, the standard design for quantum
lambda calculi has kept the quantum state outside the program: terms manipulate
pointers to an external register. This is largely due to the difficulty of
eliminating tensor products. For example, the calculus $\qlambdens$ embeds
density matrices directly within terms, where terms carry the entire
computation state, a feature particularly appealing for program verification.
However, lacking a tensor elimination construct, it can neither access the
individual qubits of a compound state nor discard them. Borgna showed that this
inability to discard qubits makes the calculus strictly less expressive than
the quantum lambda calculus of Selinger and Valiron.

In this paper we show that tensor elimination is possible in this setting. The
key observation is that the Pauli decomposition, combined with the spectral
decomposition of the Pauli matrices, allows any $n$-qubit density matrix to be
expressed as a real linear combination of tensor products of single-qubit
density matrices. Exploiting this fact, we extend $\qlambdens$ with a construct
$\mathsf{let}\ x^{\otimes n} = \rho\ \mathsf{in}\ t$, which binds each $x_i$ to
a single-qubit density matrix arising from the decomposition of $\rho$.

We equip the extended calculus with a rewrite system, a type system, and a
denotational semantics, and prove Subject Reduction, Progress, Strong
Normalisation, Soundness, and Adequacy. The new construct also recovers the
missing ability to discard qubits, thereby restoring expressiveness. Moreover,
we show that this is achieved in a physically principled way: a variable unused
in $t$ is interpreted exactly as being partial-traced out, as dictated by the
no-deleting theorem. We illustrate the resulting compositionality through
quantum teleportation and the three-qubit bit-flip code.

\end{abstract}

\section{Introduction}
\label{sec:introduction}

Quantum programming languages provide a hardware-independent, high-level
description of quantum algorithms. Since the seminal work of Selinger and
Valiron~\cite{SelingerValironMSCS06}, the standard design for higher-order
quantum languages has followed the quantum data / classical control paradigm:
the quantum state is kept in an external register, while terms manipulate names
(pointers) to the qubits stored in it. This design sidesteps the problem of
decomposing entangled compound states within the language itself. By contrast,
calculi that represent quantum states directly as terms must confront this
problem explicitly.

Another line of quantum lambda calculi takes a different route, embedding the
computation state directly as a constant within terms, instead of keeping it in
an external store. Reasoning in such calculi more directly reflects the
evolution of the computation state, which makes them especially appealing for
program verification~\cite{AvanziniDiazcaroHainryPechoux25}: since a term
carries its entire computation state, program equivalences can be established
by reasoning within the calculus itself, rather than about configurations of an
external store. One such
calculus is $\qlambdens$~\cite{DiazCaro17}, which uses density matrices to
represent states as part of terms.

This internal-state approach has recently been used as the basis for an
expectation-based analysis of higher-order quantum programs
\cite{AvanziniDiazcaroHainryPechoux25}. That work considers a derivative
calculus using vectors instead of density matrices and adding recursion, among
other extensions, and exploits its similarities with a probabilistic classical
language~\cite{AvanziniBartheDalLagoICFP21}. In that setting, the calculus
serves as an intermediate language: a higher-order quantum program can be
compiled into it, and the analysis can then be performed there.

The need for such an intermediate-language role is partly due to a limitation
of $\qlambdens$ and its derivative calculi: they lack a tensor elimination
construct, or let-tensor. As a consequence, although compound quantum states
can be represented, they cannot be decomposed within the language itself. The
difficulty is not specific to $\qlambdens$: being unable to split the state of
a multi-qubit system into its components, precisely because of entanglement, is
a limitation shared by other calculi for quantum computation, such as those of
Hasuo and Hoshino~\cite{HasuoHoshino17} and Delbecque~\cite{Delbecque08}. This
affects compositionality. A single quantum computation can still be expressed
by encoding operations directly at the matrix level, so the language can still
serve its purpose as an intermediate language. However, the individual qubits
of a compound state cannot be addressed independently.

Consider a concrete scenario: given a two-qubit state $\rho$, suppose that we
want to apply a function $f$ to the first qubit and a function $g$ to the
second qubit independently. In $\qlambdens$, there is no way to express this as
$f(x_1)$ and $g(x_2)$, where $x_1$ and $x_2$ are the individual qubits of
$\rho$. One must instead encode the operation at the matrix level, using an
operator such as $U_f \otimes U_g$ applied to the whole state, losing the
modularity that lambda calculi are designed to provide.

This limitation also has an expressiveness-level counterpart. Since
$\qlambdens$ provides no mechanism for decomposing a state, it also provides no
mechanism for discarding qubits. Borgna~\cite{Borgna19} showed that, as a
consequence, some quantum computations cannot be directly encoded in it.

In this paper we present an extension of $\qlambdens$ that addresses both
concerns. It provides compositional access to the individual qubits of a state
and, through affine discard (an \emph{affine} type system being one that
admits weakening but not contraction), recovers the ability to discard qubits,
thereby restoring expressiveness. Moreover, the resulting notion of discard is
shown to coincide exactly with the partial trace, yielding a physically
meaningful interpretation of affine weakening.

The key observation is that combining the Pauli decomposition of density
matrices with the spectral decomposition of the Pauli
matrices~\cite{NielsenChuang10} allows any $n$-qubit density matrix to be
expressed as a real linear combination of tensor products of single-qubit
density matrices. We exploit this fact to extend $\qlambdens$ with a construct
$\letxn{x}{\rho}{t}$ that binds the individual qubit components of a compound
state within a body, so that programs may operate on them in an algebraically
coherent way.

\paragraph{Contributions.}
We make the following contributions.
\begin{enumerate}
\item \textbf{A compositional $\mathsf{let}$ construct.}
We extend $\qlambdens$ with a construct $\letxn{x}{\rho}{t}$ that binds the
variables $x_1, \dots, x_n$ to the single-qubit density matrices arising from
the decomposition of $\rho$ within the body $t$. The decomposition is
obtained by combining the Pauli decomposition with the spectral decomposition
of the Pauli matrices: the state $\rho$ is expressed as a weighted sum of
tensor products of single-qubit density matrices, and the body $t$ is
evaluated for each such tensor product, weighted accordingly.

\item \textbf{A denotational semantics.}
We provide a denotational semantics for the extended calculus, interpreting
the new $\mathsf{let}$ construct through the corresponding decomposition of
density matrices.

\item \textbf{Metatheoretical properties.}
We prove Subject Reduction, Progress, Strong Normalisation, Soundness, and
Adequacy for the extended calculus. We also prove the central result specific
to this extension: discard coincides with partial trace. This establishes
that the $\mathsf{let}$ construct is faithful to the no-deleting theorem.

\item \textbf{Illustrative examples.}
We implement quantum teleportation and the three-qubit bit-flip error
correction code in the extended calculus. Both examples demonstrate how the
$\mathsf{let}$ construct enables compositional programming over multi-qubit
states, including the separation of roles in teleportation and ancilla
management in error correction. Throughout the paper, we also use the Bell
state as a running example: we exhibit its Pauli and combined decompositions
and use it to show how discarding a qubit computes the partial trace.
\end{enumerate}

\section{Preliminaries}
\label{sec:preliminaries}

We briefly recall quantum computing in the density matrix formalism and the
base calculus $\qlambdens$. We refer to \cite{NielsenChuang10} for a
comprehensive introduction to the former.

\subsection{Quantum Computing with Density Matrices}
\label{sec:density-matrices}

While a pure state can be described by a unit vector in a Hilbert space, this
representation cannot capture situations in which our knowledge of the system
is incomplete, for instance after a measurement whose outcome is unknown. The
density matrix formalism encompasses both pure and mixed states.

\begin{definition}[Density matrix]
A \emph{density matrix} (or density operator) is a positive semidefinite
matrix $\rho$ with $\tr(\rho)=1$, where $\tr$ denotes the trace. For an
$n$-qubit system, $\rho$ is a $2^n\times 2^n$ matrix. We write $\D$ for the
set of density matrices of $n$-qubit systems.
\end{definition}

A pure state $\ket{\psi}$ corresponds to the density matrix
$\rho=\ketbra{\psi}{\psi}$. A mixed state, representing a statistical ensemble
in which the system is in state $\rho_i$ with probability $p_i$ (where
$\sum_i p_i=1$), is described by $\rho=\sum_i p_i\rho_i$.

The four postulates of quantum mechanics can be formulated entirely in terms of
density matrices. The \emph{evolution} of a closed system is described by a
unitary operator $U$, which transforms a state $\rho$ into
$U\rho U^\dagger$. A \emph{measurement} is described by a family of operators
$\{\pi_i\}_i$ satisfying $\sum_i \pi_i^\dagger\pi_i=I$. If the system is in
state $\rho$, the probability of outcome $i$ is
$p_i=\tr(\pi_i^\dagger\pi_i\,\rho)$ and, when $p_i>0$, the corresponding
post-measurement state is
$\pi_i\rho\pi_i^\dagger/p_i$. Finally, composite systems are represented by
tensor products: if systems $A$ and $B$ are respectively in states $\rho^A$
and $\rho^B$, then their joint state is $\rho^A\otimes\rho^B$.

\begin{example}[Bell state]
\label{ex:bell-def}
The Bell state
$\ket{\beta_{00}}=\frac{1}{\sqrt{2}}(\ket{00}+\ket{11})$
has density matrix
\[
\beta_{00}
=
\frac12(
\ketbra{00}{00}
+
\ketbra{00}{11}
+
\ketbra{11}{00}
+
\ketbra{11}{11}
).
\]
This state cannot be written as $\rho^A\otimes\rho^B$ for any single-qubit
density matrices $\rho^A$ and $\rho^B$. It will serve as a running example
throughout the paper.
\end{example}

Two physical principles constrain any language manipulating quantum data.
Quantum information cannot be freely \emph{duplicated} (the \emph{no-cloning}
theorem~\cite{WoottersZurekN82}), which motivates the affine type system of
$\qlambdens$. Likewise, quantum information cannot be freely \emph{deleted}
(the \emph{no-deleting} theorem~\cite{PatiBraunstein00}). We return to this
point in Section~\ref{sec:discard}, where we show that discard corresponds
exactly to partial trace.

\subsection{The Calculus \texorpdfstring{$\qlambdens$}{lambda-rho-circ}}
\label{sec:qlambdens}

The calculus $\qlambdens$~\cite{DiazCaro17} is a quantum extension of the
lambda calculus in which quantum data is represented by density matrices.
Following the quantum data / classical control paradigm, measurements are
handled through deterministic rewriting with probabilistic control: a
measurement does not reduce to a single outcome but rather to a probabilistic
mixture containing all possible outcomes.

The grammar of terms (Table~\ref{tab:Grammar}) contains the usual lambda
calculus constructs together with terms corresponding to the quantum
operations described in the previous section and constructs for probabilistic
control. In $\rho^n$, the superscript $n$ denotes the number of qubits.
The term $U^m t$ applies a unitary transformation, $\pi^m t$ performs a
measurement, and $t\otimes r$ forms a composite system. The term
$\sum_i p_i t_i$ represents a probabilistic mixture. Finally,
$\mathsf{letcase}^\circ$ eliminates a measurement result by producing the
weighted sum of all possible branches.

Although $\qlambdens$ provides a constructor for composite systems through
$\otimes$, it provides no corresponding elimination construct. This absence
prevents the individual components of a compound state from being manipulated
independently and is precisely the limitation addressed in the present work.

\begin{table}[t]
  \centering
  \begin{align*}
    t & := x\mid \lambda x.t \mid tt & \text{(Standard lambda calculus)}\\
      & \hspace{2mm}\mid \rho^n\mid U^n t\mid \pi^n t\mid t\otimes t & \text{(Quantum postulates)}\\
      & \hspace{2mm}\mid \sum_{i=1}^n p_i t_i\mid  \qletcase xt{t,\dots,t} & \text{(Probabilistic control)}
  \end{align*}
  where $p_i\in (0,1]$, $\sum_{i=1}^n p_i = 1$, and $\sum$ is considered modulo
  associativity and commutativity.
  \caption{Grammar of terms of $\qlambdens$.}
  \label{tab:Grammar}
\end{table}

To state Progress, we distinguish values from reducible terms.
Since measurements are evaluated through $\mathsf{letcase}^\circ$, a term
$\pi^m\rho^n$ is considered a value.
\begin{definition}[Values]
  \label{def:values}
  A value is a term $v$ defined by the following grammar
  \begin{align*}
  \hat w &:= \lambda x.t \mid \pi^m \rho^n \mid \sum_i p_i \hat w_i\\
  v &:= \hat w\mid\rho^n
  \end{align*}
  where $\forall i\neq j, \hat w_i \neq \hat w_j$. This condition is to be read
  syntactically: the summands must be pairwise distinct \emph{as terms} (up to
  renaming of bound variables); in particular, two $\lambda$-abstractions are
  identified only when they are the same term, regardless of the function they
  denote.
\end{definition}

The rewrite system (Table~\ref{tab:TRS}) uses a deterministic relation
$\rightarrow$. We write $t[r/x]$ for the usual capture-avoiding substitution of
$r$ for the free occurrences of $x$ in $t$, following Barendregt's variable
convention. Since density matrices are themselves terms, several rewrite rules
relate the syntactic layer to the underlying mathematical model. To keep the
two apart, within the rewrite rules we write $\hotimes$ and $\hsum$ for the
tensor product and the linear combination \emph{of matrices}, the usual
operations of linear algebra, reserving $\otimes$ and $\sum$ for the
homonymous term constructors. The overline notation extends an operator to the
dimension of the state to which it is applied. For instance, when applying a
two-qubit unitary $U^2$ to a three-qubit state, $\ext{U^2}=U^2\hotimes I^1$.
Notably, measurement does not reduce by itself; it only reduces when it is the
argument of a $\mathsf{letcase}^\circ$, producing a weighted sum over all
possible outcomes.
We include a distributivity rule for
$\mathsf{letcase}^\circ$ over probabilistic mixtures.
This rule is needed to formulate Progress using the notion of value above.
We write $t \downarrow r$ when there exists a term $u$ such that
$t \rightarrow^* u$ and $r \rightarrow^* u$.

\begin{table}[t]
  \centering
  \[
    \begin{array}{r@{}l@{\qquad}l}
      (\lambda x.t)r &\rightarrow t[r/x]\\
      \qletcase x{\pi^m\rho^n}{t_0,\dots,t_{2^m-1}} &\rightarrow\sum\limits_i p_i t_i[\rho^n_i/x]
      & \textrm{with }\left\{ \begin{array}{l}
	  \rho_i^n=\frac{\ext{\pi_i}\rho^n\ext{\pi_i}^\dagger}{p_i}\\
	  p_i=\tr(\ext{\pi_i}^\dagger\ext{\pi_i}\rho^n)
      \end{array}\right.\\
      U^m\rho^n &\rightarrow {\rho'}^n &\textrm{with }\ext{U^m}\rho^n\ext{U^m}^\dagger={\rho'}^n\\
      \rho\otimes\rho'&\rightarrow \rho'' & \text{with }\rho''=\rho\hotimes\rho'\\
      \sum_i p_i\rho_i&\rightarrow \rho' & \text{with }\rho'=\hsum_ip_i\rho_i\\
      \sum_i p_i t &\rightarrow t\\
      (\sum_i p_i t_i)r &\rightarrow \sum_i p_i (t_ir)\\
      \qletcase{x}{\sum_i p_i \hat w_i}{t_0,\dots,t_{2^m-1}}
      &\multicolumn{2}{l}{\rightarrow \sum_i p_i\,\qletcase{x}{\hat w_i}{t_0,\dots,t_{2^m-1}}}
    \end{array}
  \]
  \[
    \infer{ts\rightarrow rs}{t\rightarrow r}\quad \infer{st\rightarrow sr}{t\rightarrow r}\quad
    \infer{U^nt\rightarrow U^nr}{t\rightarrow r}\quad \infer{\pi^n t\rightarrow \pi^n r}{t\rightarrow r}\quad
    \infer{t\otimes s\rightarrow r\otimes s}{t\rightarrow r}\quad \infer{s\otimes t\rightarrow s\otimes r}{t\rightarrow r}
  \]
  \[
    \infer[\scriptstyle(\forall i\neq j, t_i=r_i)]{\sum_{i=1}^n p_it_i\rightarrow
      \sum_{i=1}^n p_ir_i} {t_j\rightarrow r_j}
    \qquad
    \infer {\qletcase xt{\vec s}\rightarrow\qletcase xr{\vec s}} {t\rightarrow r}
  \]
  \caption{Rewrite system of $\qlambdens$. In the last rule $\vec s$ abbreviates $s_0,\dots,s_{2^m-1}$.}
  \label{tab:TRS}
\end{table}

The type system (Table~\ref{tab:TS}) is affine: variables can be used at most
once, preventing the duplication of quantum states (respecting the no-cloning
theorem). In rules that combine two contexts $\Gamma$ and $\Delta$ (namely
$\multimap_e$ and $\otimes$), the two contexts are required to be
\emph{disjoint}, which is what enforces linearity. The type $n$ denotes an
$n$-qubit density matrix, the type $(m,n)$ denotes the result of measuring $m$
out of $n$ qubits and $A \multimap B$ is the linear function type. While
composite systems can be introduced through the $\otimes$ rule, there is no
corresponding elimination rule.

\begin{table}[t]
  \centering
  \[
    A:= n\mid (m,n)\mid A\multimap A \qquad\text{where } m\leq n\in\mathbb N.
  \]
  \[
    \infer[\mathsf{ax}]{\Gamma,x:A\vdash x:A}{} \quad
    \infer[\multimap_i]{\Gamma\vdash\lambda x.t:A\multimap B}{\Gamma,x:A\vdash
      t:B} \quad \infer[\multimap_e]{\Gamma,\Delta\vdash tr:B}{\Gamma\vdash
      t:A\multimap B & \Delta\vdash r:A}
  \]
  \[
    \infer[\mathsf{ax}_\rho]{\Gamma\vdash\rho^n:n}{} \quad
    \infer[\mathsf{u}]{\Gamma\vdash U^mt:n}{\Gamma\vdash t:n} \quad
    \infer[\mathsf{m}]{\Gamma\vdash\pi^mt:(m,n)}{\Gamma\vdash t:n} \quad
    \infer[\otimes]{\Gamma,\Delta\vdash t\otimes r:n+m}{\Gamma\vdash t:n &
      \Delta\vdash r:m}
  \]
  \[
    \infer[\mathsf{lc}] {\Gamma\vdash\qletcase xr{t_0,\dots,t_{2^m-1}}:A}
    {x:n\vdash t_0:A & \dots & x:n\vdash t_{2^m-1}:A & \Gamma\vdash r:(m,n) }
  \]
  \[
    \infer[+]{\Gamma\vdash\textstyle\sum_{i=1}^n p_i t_i:A} { \Gamma\vdash t_1:A &\dots&
      \Gamma\vdash t_n:A & \sum_{i=1}^np_i=1 }
  \]
  \caption{Type system for $\qlambdens$.}
  \label{tab:TS}
\end{table}

\section{The Compositional Extension}
\label{sec:extension}

We first develop the Pauli spectral decomposition machinery needed for the new
construct, then present the extended calculus.

\subsection{Pauli and Spectral Decompositions}
\label{sec:pauli}

The extension relies on combining the Pauli decomposition of density matrices
with the spectral decomposition of the Pauli operators. We begin by recalling
both constructions.
The Pauli matrices, together with the
identity, form a basis for the space of $2\times 2$ Hermitian
matrices~\cite[\S2.1.3]{NielsenChuang10}:
\[
  I = \begin{pmatrix} 1 & 0 \\ 0 & 1 \end{pmatrix}, \;
  X = \begin{pmatrix} 0 & 1 \\ 1 & 0 \end{pmatrix}, \;
  Y = \begin{pmatrix} 0 & -i \\ i & 0 \end{pmatrix}, \;
  Z = \begin{pmatrix} 1 & 0 \\ 0 & -1 \end{pmatrix}.
\]

\begin{definition}[Pauli operator basis]
  \label{def:pauli-basis}
  Let $n \in \mathbb{N}^*$. The Pauli operator basis of size $n$ is the set
  $\mathcal{P}_n = \{ \bigotimes_{i=1}^n M_i \mid M_i \in \{I, X, Y, Z\}\}$.
\end{definition}

The set $\mathcal{P}_n$ contains $4^n$ elements and forms an orthogonal basis of
$\mathbb{C}^{2^n \times 2^n}$ with respect to the Hilbert--Schmidt inner product
$\langle A, B \rangle = \tr(A^\dagger B)$~\cite{KoskaBaboulinGazda24}. Hence any
density matrix $\rho$ can be written as
\begin{equation}\label{eq:pauli-decomp}
  \rho = \sum_{P_i \in \mathcal{P}_n} \alpha_{i} P_i, \qquad
  \alpha_{i} = \tfrac{1}{2^n} \tr\!\big( P_i \rho \big),
\end{equation}
with $\alpha_{i} \in \mathbb{R}$ (real because $\rho$ and each $P_i$ are
Hermitian).

\begin{example}[Pauli decomposition of the Bell state]
  \label{ex:bell-pauli}
  Since $n = 2$, the Pauli basis $\mathcal{P}_2$ has $16$ elements
  $M_1 \otimes M_2$. The only non-zero coefficients
  are
  \begin{equation}\label{eq:bell-alphas}
    \alpha_{II} = \tfrac{1}{4},\quad
    \alpha_{XX} = \tfrac{1}{4},\quad
    \alpha_{YY} = -\tfrac{1}{4},\quad
    \alpha_{ZZ} = \tfrac{1}{4}.
  \end{equation}
  The corresponding Pauli decomposition is therefore
  \begin{equation}\label{eq:bell-pauli}
    \beta_{00} = \tfrac{1}{4}(I \otimes I + X \otimes X - Y \otimes Y + Z \otimes Z).
  \end{equation}
\end{example}

To obtain a decomposition into tensor products of single-qubit density
matrices, we must further decompose the Pauli operators appearing in
Equation~\eqref{eq:pauli-decomp}. For this purpose we use the
spectral decomposition theorem.

\begin{theorem}[Spectral decomposition]
  \label{thm:spectral}
  Let $A$ be a normal matrix of dimension $d \times d$. Then $A = \sum_{i=1}^{d}
  \lambda_i \ketbra{v_i}{v_i}$, where $\lambda_1, \dots, \lambda_d$ are the
  eigenvalues of $A$ (with multiplicity) and $\ket{v_1}, \dots, \ket{v_d}$ form
  an orthonormal basis of eigenvectors.
\end{theorem}

\subsection{Spectral Decomposition of the Pauli Matrices}
\label{sec:pauli-spectral}

By Theorem~\ref{thm:spectral}, each Pauli matrix can be written as a linear
combination of its eigenprojectors. For the Pauli matrices, these
decompositions are
\begin{alignat*}{2}
  I &= \ketbra{0}{0} + \ketbra{1}{1}, &\qquad
  Z &= \ketbra{0}{0} - \ketbra{1}{1}, \\
  X &= \ketbra{+}{+} - \ketbra{-}{-}, &\qquad
  Y &= \ketbra{i}{i} - \ketbra{-i}{-i},
\end{alignat*}
where
\[
  \ket{\pm} = \frac{1}{\sqrt{2}}(\ket{0} \pm \ket{1}),
  \qquad
  \ket{\pm i} = \frac{1}{\sqrt{2}}(\ket{0} \pm i\ket{1}).
\]

Observe that every projector appearing above is a valid single-qubit density
matrix. More generally, each Pauli matrix $M$ admits a decomposition
\[
  M
  =
  \lambda_{1} \rho_{1}
  +
  \lambda_{2} \rho_{2},
\]
where $\rho_1,\rho_2\in\D[1]$ are rank-one density matrices and
$\lambda_1,\lambda_2\in\{-1,+1\}$ are the corresponding eigenvalues.

\subsection{Combined Decomposition}
\label{sec:combined}
We now combine the Pauli decomposition with the spectral decomposition of the
Pauli matrices.

\begin{proposition}[Decomposition into single-qubit density matrices]
  \label{prop:combined}
  Every density matrix $\rho\in\D[n]$ admits a decomposition of the form
  \begin{equation}
    \rho
    =
    \sum_{i=1}^{4^n}
    \sum_{\vec l\in\{1,2\}^n}
    p_{i\vec l}
    \bigotimes_{k=1}^n \rho_{i_k}^{l_k},
    \qquad
    p_{i\vec l}
    =
    \alpha_i
    \prod_{k=1}^n \lambda_{i_k}^{l_k},
    \label{eq:full-decomp}
  \end{equation}
  where $\alpha_i$ is the coefficient of the $i$-th Pauli operator
  $P_i=M_{i_1}\otimes\cdots\otimes M_{i_n}$ in the Pauli
  decomposition~\eqref{eq:pauli-decomp} of $\rho$, and $\lambda_{i_k}^{l_k}$
  and $\rho_{i_k}^{l_k}$ are, respectively, the eigenvalues in $\{-1,+1\}$ and
  the corresponding rank-one eigenprojectors of the $k$-th factor $M_{i_k}$
  of $P_i$, as given in Section~\ref{sec:pauli-spectral}. In particular, each
  $\rho_{i_k}^{l_k}\in\D[1]$ is a single-qubit density matrix.
\end{proposition}

\begin{proof}
  Let
  \[
    P_i=M_{i_1}\otimes\cdots\otimes M_{i_n}
  \]
  be an element of the Pauli basis. By the spectral decomposition of each
  factor,
  \[
    M_{i_k}
    =
    \lambda_{i_k}^1\rho_{i_k}^1
    +
    \lambda_{i_k}^2\rho_{i_k}^2.
  \]
  Expanding the tensor product and substituting the resulting expression into
  Equation~\eqref{eq:pauli-decomp} yields the claimed decomposition.
  \qed
\end{proof}

The coefficients $p_{i\vec l}$ are not necessarily non-negative, so
Equation~\eqref{eq:full-decomp} is generally not a convex
combination.  Nevertheless, evaluating the trace of
Equation~\eqref{eq:full-decomp} yields
\[
  \sum_{i,\vec l} p_{i\vec l}
  =
  \tr(\rho)
  =
  1.
\]
Indeed, each tensor product $\bigotimes_k \rho_{i_k}^{l_k}$ has trace one.  The
individual terms $\bigotimes_k \rho_{i_k}^{l_k}$ are genuine density matrices;
the non-classical behaviour of entangled states is encoded in the coefficients,
which may be negative.

The decomposition into single-qubit density matrices is the key mathematical
tool underlying our $\mathsf{let}$ construct.  It provides a canonical way to
represent any multi-qubit state as a real linear combination of tensor products
of single-qubit density matrices, allowing the individual components to be
manipulated independently within a program.

\begin{example}[combined decomposition of the Bell state]
  \label{ex:bell-combined}
  Continuing from Example~\ref{ex:bell-pauli}, only the four pairs $II$, $XX$,
  $YY$, $ZZ$ carry non-zero $\alpha_i$ (Equation~\eqref{eq:bell-alphas}), so the sum has
  $4 \times 4 = 16$ non-zero terms. Expanding each, the full combined
  decomposition is
  \begin{align}
    &\beta_{00}\nonumber\\
    &=\tfrac{1}{4}\ketbra{0}{0} \otimes \ketbra{0}{0}
    + \tfrac{1}{4}\ketbra{0}{0} \otimes \ketbra{1}{1}
    + \tfrac{1}{4}\ketbra{1}{1} \otimes \ketbra{0}{0}
    + \tfrac{1}{4}\ketbra{1}{1} \otimes \ketbra{1}{1} \nonumber\\
    &+ \tfrac{1}{4}\ketbra{+}{+} \otimes \ketbra{+}{+}
    - \tfrac{1}{4}\ketbra{+}{+} \otimes \ketbra{-}{-}
    - \tfrac{1}{4}\ketbra{-}{-} \otimes \ketbra{+}{+}
    + \tfrac{1}{4}\ketbra{-}{-} \otimes \ketbra{-}{-} \nonumber\\
    &- \tfrac{1}{4}\ketbra{i}{i} \otimes \ketbra{i}{i}
    + \tfrac{1}{4}\ketbra{i}{i} \otimes \ketbra{-i}{-i}
    + \tfrac{1}{4}\ketbra{-i}{-i} \otimes \ketbra{i}{i}
    - \tfrac{1}{4}\ketbra{-i}{-i} \otimes \ketbra{-i}{-i} \nonumber\\
    &+ \tfrac{1}{4}\ketbra{0}{0} \otimes \ketbra{0}{0}
    - \tfrac{1}{4}\ketbra{0}{0} \otimes \ketbra{1}{1}
    - \tfrac{1}{4}\ketbra{1}{1} \otimes \ketbra{0}{0}
    + \tfrac{1}{4}\ketbra{1}{1} \otimes \ketbra{1}{1}.
    \label{eq:bell-full-decomp}
  \end{align}
  Some tensor products appear more than once in the expansion because they arise
  from different Pauli terms.
  We deliberately keep the decomposition in its
  expanded form to make the contribution of each Pauli component explicit.

  Observe that all terms are tensor products of single-qubit density matrices,
  while the correlations of $\beta_{00}$ are reflected in the pattern of
  positive and negative coefficients.
\end{example}

\begin{remark}[Size of the decomposition]\label{rmk:size}
  The sum in Equation~\eqref{eq:full-decomp} may, in the worst
  case, contain up to $8^n$ summands, although only those Pauli
  strings $P_i$ with non-zero coefficient $\alpha_i$ contribute --- for the Bell
  state, only $4$ of the $16$ Pauli pairs are non-zero, giving the $16$ terms of
  Example~\ref{ex:bell-combined}. This exponential growth is not introduced by the $\mathsf{let}$
  construct: an $n$-qubit density matrix is already represented by a $2^n\times
  2^n$ matrix, and the decomposition merely re-expresses this existing
  information as a sum over single-qubit components rather than as the entries
  of a single $2^n \times 2^n$ matrix.
  This behaviour is therefore inherent to the representation of quantum states
  rather than a consequence of the $\mathsf{let}$ construct itself.
  $\lambda^{\circ}_{\rho}$ is intended as a calculus for reasoning
  about and verifying quantum programs, where a term carries its computation
  state for the purposes of equational and denotational reasoning, rather than as
  an executable model for efficient classical simulation;
  a faithful classical representation of an entangled $n$-qubit state may incur
  an exponential cost, with or without the $\mathsf{let}$ construct.
\end{remark}

\subsection{Extended Grammar, Rewrite System, and Type System}
\label{sec:extended}

We now introduce the construct enabled by the decomposition developed in the
previous subsections. We extend the grammar of $\qlambdens$ with a new term
$\letxn{x}{\rho}{t}$, which binds $n$ variables $x_1, \dots, x_n$, each of type
$1$ (single-qubit density matrix), within the body $t$. The superscript
$\otimes n$ in $x^{\otimes n}$ indicates that $x$ is being decomposed into $n$
tensor components.

We also generalise the construct $\sum_i p_i t_i$ from probabilistic mixtures
to finite real linear combinations whose coefficients sum to one, but are not
required to be positive. This extension is necessary because the decomposition
of Proposition~\ref{prop:combined} may involve negative coefficients.

The new reduction rules for the $\mathsf{let}$ construct are given in
Table~\ref{tab:LetTRS}.

\begin{table}[t]
  \centering
  \[
    \letxn{x}{\rho^n}{s} \rightarrow \sum_{i=1}^{4^n} \sum_{\vec{l} \in \{1,2\}^n} p_{i\vec{l}} \, s[\rho_{i_{1}}^{l_1}/x_1, \cdots, \rho_{i_{n}}^{l_n}/x_n]
  \]
  \[
    \infer {\letxn{x}{t}{s} \rightarrow \letxn{x}{r}{s}} {t\rightarrow r}
  \]
  where $\rho = \hsum_{i=1}^{4^n} \hsum_{\vec{l} \in \{1,2\}^n} p_{i\vec{l}} \hbigotimes_{k=1}^n \rho_{i_k}^{l_k}$ as in Proposition~\ref{prop:combined}.
  \caption{Rewrite rules for the $\mathsf{let}$ construct.}
  \label{tab:LetTRS}
\end{table}

The reduction works as follows: when the source term is a density matrix $\rho^n$,
its decomposition given by Proposition~\ref{prop:combined} is computed.
Each tensor product $\bigotimes_k\rho_{i_k}^{l_k}$ induces a simultaneous
substitution of the components $\rho_{i_k}^{l_k}$ for the variables $x_k$ in
the body. The results are then recombined according to the coefficients
$p_{i\vec l}$.

The type system is extended with a single new rule:
\[
  \infer[\mathsf{let}]{\Gamma, \Delta \vdash \mathsf{let}\ x^{\otimes n} = t\ \mathsf{in}\ s : A}
  {\Gamma \vdash t : n \quad \Delta, x_1:1, \ldots, x_n:1 \vdash s : A}
\]
The rule requires that $t$ has type $n$, and types the body $s$ under
assumptions corresponding to the $n$ single-qubit components of the
decomposition. The contexts $\Gamma$ and $\Delta$ are split, maintaining the
affine discipline.

\section{Denotational Semantics}
\label{sec:semantics}

A \emph{completely positive map} (CPM) is a linear map between spaces of
operators that preserves positivity under extension by arbitrary ancillary
systems. CPMs are the standard model for physically realisable quantum
operations~\cite[Ch.~8]{NielsenChuang10}. We write
$\CPM(\D, \D[m])$ for the set of CPMs from $\D$ to
$\D[m]$. Types are interpreted as
\begin{align*}
  \tsem{n} &= \D[n] &
  \tsem{(m,n)} &= \D[n] &
  \tsem{A \multimap B} &= \CPM(\tsem{A}, \tsem{B}).
\end{align*}

The interpretation is defined relative to a typing judgement
$\Gamma \vdash t : A$ and a valuation $\theta \vDash \Gamma$, that is,
$\theta(x)\in\tsem{A}$ for every assumption $x:A\in\Gamma$.
Given such a valuation, the interpretation of $\qlambdens$ terms is
\begin{align*}
  \fsem{x} &= \theta(x)
  \\
  \fsem{\lambda x.t} &= \rho \mapsto \fsem[\theta,x\mapsto\rho]{t}
  \\
  \fsem{t\ r} &= \fsem{t}(\fsem{r})
  \\
  \fsem{\rho^n} &= \rho^n
  \\
  \fsem{U^m t} &= \ext{U^m}\,\fsem{t}\,\ext{U^m}^{\dagger}
  \\
  \fsem{\pi^m t} &= \sum_{i=0}^{2^m-1} \ext{\pi_i}\,\fsem{t}\,\ext{\pi_i}^{\dagger}
  \\
  \fsem{t\otimes r} &= \fsem{t}\otimes\fsem{r}
  \\
  \fsem{\sum_i p_i t_i} &= \sum_i p_i\,\fsem{t_i}
  \\
  \fsem{\qletcase{x}{r}{t_0,\ldots,t_{2^m-1}}} &= \sum_{i=0}^{2^m-1} \tr(\ext{\pi_i}^{\dagger}\ext{\pi_i}\,\fsem{r})\, \fsem[\theta,x\mapsto\frac{\ext{\pi_i}\,\fsem{r}\,\ext{\pi_i}^{\dagger}}{\tr(\ext{\pi_i}^{\dagger}\ext{\pi_i}\,\fsem{r})}]{t_i}
  \\
  \fsem{\letxn{x}{t}{s}}
  &=
  \sum_{i=1}^{4^n}
  \sum_{\vec l\in\{1,2\}^n}
  p_{i\vec l}\,
  \fsem[\theta,
        x_1\mapsto\rho_{i_1}^{l_1},
        \ldots,
        x_n\mapsto\rho_{i_n}^{l_n}]{s}.
\end{align*}
where, for the interpretation of the $\mathsf{let}$ construct,
$\fsem{t}$ is decomposed as in Proposition~\ref{prop:combined} as
\[
  \fsem{t}
  =
  \sum_{i=1}^{4^n}
  \sum_{\vec l\in\{1,2\}^n}
  p_{i\vec l}
  \bigotimes_{k=1}^n \rho_{i_k}^{l_k}
\]

This interpretation is well defined, since every density matrix admits such a
decomposition.

\section{Properties}
\label{sec:properties}

We show that discarding a qubit through the affine type system with the
$\mathsf{let}$ construct is equivalent to tracing it out, and that the extended
calculus preserves the fundamental properties of Subject Reduction, Progress,
Strong Normalisation, Soundness, and Adequacy.
Full proofs are given in the appendix.

\subsection{Discard as Partial Trace}
\label{sec:discard}

We first recall the partial trace, the physical operation of discarding a
subsystem. The \emph{partial trace over $A$} is the unique linear map $\tr_A$
defined on product operators by $\tr_A(\ketbra{a_1}{a_2} \otimes
\ketbra{b_1}{b_2}) = \braket{a_2}{a_1}\, \ketbra{b_1}{b_2}$; the partial trace
over subsystem $k$ acts as $\tr$ on the $k$-th tensor factor and as the identity
on all others~\cite[\S2.4.3]{NielsenChuang10}. On product states
$\tr_A(\rho^A \otimes \rho^B) = \rho^B$, but for entangled states it produces a
mixed state even when the global state is pure.

\begin{example}[partial trace of the Bell state]
  \label{ex:bell-tr}
  Tracing out the first qubit of the Bell state $\beta_{00}$ gives the
  maximally mixed state, reflecting that no information about the second qubit
  can be obtained without access to the first:
  \[
    \tr_1(\beta_{00})
    = \tfrac{1}{2}(\ketbra{0}{0} + 0 + 0 + \ketbra{1}{1})
    = \tfrac{1}{2}\ketbra{0}{0} + \tfrac{1}{2}\ketbra{1}{1}
    = \tfrac{I}{2}.
  \]
\end{example}

The relevance of the partial trace here is physical. Just as quantum
information cannot be cloned, it cannot be deleted~\cite{PatiBraunstein00}.

Any operation that appears to delete a qubit must, in reality, hide its
information somewhere else (e.g.\ in the environment); the physical operation
corresponding to \emph{ignoring} a subsystem is precisely the partial trace,
which marginalises over it. A language construct that ``discards'' a qubit
must therefore, for physical faithfulness, denote a partial trace rather than a
literal deletion. We now show that the $\mathsf{let}$ construct does exactly
this: \emph{ignoring} a component in the body (permitted by the affine type
system) is not merely a syntactic convenience, but is computationally
equivalent to first computing the reduced density matrix of the remaining qubits
via the partial trace, and then decomposing that reduced state.

\begin{proposition}[Discard as Partial Trace]
  \label{prop:discard}
  Let $\rho^n$ be an $n$-qubit density matrix with $n \geq 2$, and let $s$ be a
  term with $x_2:1,\ldots,x_n:1 \vdash s : A$ (so, by affinity, $x_1 \notin
  FV(s)$). Then
  \[
    \letxn{x}{\rho^n}{s}
    \;\mathrel{\downarrow}\;
    \letxk{y}{n-1}{\tr_1(\rho^n)}{s[y_1/x_2, \ldots, y_{n-1}/x_n]}.
  \]
  The same holds for discarding any other qubit position $k$ by relabeling.
\end{proposition}
\begin{proof}
  (\emph{Sketch.}) Let $\rho = \sum_{i,\vec{l}} p_{i\vec{l}}\,
  \bigotimes_{k=1}^{n} \rho_{i_k}^{l_k}$ be the
  decomposition into single-qubit density matrices given by Proposition~\ref{prop:combined}.
  Applying the main $\mathsf{let}$ rule and
  using $x_1 \notin FV(s)$, the left-hand side reduces to
  \[
    \sum_{i,\,l_2,\ldots,l_n}(\sum_{l_1} p_{i\vec{l}})\,
    s[\rho_{i_2}^{l_2}/x_2,\ldots,\rho_{i_n}^{l_n}/x_n].
  \]
  On the other hand, by
  linearity of $\tr_1$ and $\tr(\rho_{i_1}^{l_1}) = 1$,
  \[
    \tr_1(\rho^n) = \sum_{i,\,l_2,\ldots,l_n}
    \big(\textstyle\sum_{l_1} p_{i\vec{l}}\big)\;
    \bigotimes_{k=2}^{n} \rho_{i_k}^{l_k},
  \]
  which is precisely the combined decomposition of $\tr_1(\rho^n)$. Reducing the
  right-hand side with this decomposition yields the same sum. Both sides reduce
  to a common term. \qed
\end{proof}

\begin{remark}
  This connects affine variable discard to the no-deleting theorem: when a
  variable $x_k$ is not mentioned in the body of a $\mathsf{let}$, the
  reduction silently sums out its contribution exactly as the partial trace
  does in quantum mechanics. The affine type system thus enforces the
  quantum-mechanical principle that ignoring a subsystem is equivalent to
  computing its reduced state.

  In particular, affine discard is not merely admissible but has the exact
  physical interpretation prescribed by quantum mechanics.

  The property generalises to the simultaneous discard of any subset of
  qubits, since partial traces over disjoint subsystems commute.
\end{remark}

\begin{example}[discard computes the partial trace]
  \label{ex:bell-discard}
  Consider $\letxk{x}{2}{\beta_{00}}{x_2}$. Since $x_1 \notin FV(x_2)$, the
  first qubit variable is discarded by the affine type system, so we expect the
  result to be $\tr_1(\beta_{00})$. Applying the $\mathsf{let}$ rule with the
  combined decomposition~(\ref{eq:bell-full-decomp}) and keeping only the second
  tensor factor, the terms from $X \otimes X$, $Y \otimes Y$, and $Z \otimes Z$
  cancel pairwise; only $I \otimes I$ contributes, so
  \[
    \letxk{x}{2}{\beta_{00}}{x_2}
    \;\rightarrow\;
    \tfrac{1}{2}\ketbra{0}{0} + \tfrac{1}{2}\ketbra{1}{1}
    \;\rightarrow\; \tfrac{I}{2} = \tr_1(\beta_{00}),
  \]
  coinciding with Example~\ref{ex:bell-tr}, as guaranteed by
  Proposition~\ref{prop:discard}.
\end{example}

\needspace{8\baselineskip}
\subsection{Subject Reduction, Progress, and Strong Normalisation}

\begin{theorem}[Subject Reduction]
  \label{thm:SR}
  If $\Gamma\vdash t:A$ and $t\rightarrow r$, then $\Gamma\vdash r:A$.
\end{theorem}

The proof is by induction on the derivation of $t \rightarrow r$; the cases for
the original constructs follow \cite{DiazCaro17}. For the main $\mathsf{let}$
reduction, inversion gives $\Gamma \vdash \rho^n : n$ and $\Delta, x_1:1,
\ldots, x_n:1 \vdash s : A$; each $\rho_{i_{k}}^{l_k}$ is a single-qubit
density matrix, so repeated use of the substitution lemma gives $\Delta \vdash
s[\rho_{i_1}^{l_1}/x_1, \cdots] : A$,
and Proposition~\ref{prop:combined} ensures that
$\sum_{i,\vec l} p_{i\vec l}=1$, so the rule $+$ yields the result.

The $\mathsf{let}$ construct is not a value. Indeed, if its source term is a
density matrix the main reduction rule applies; otherwise, any reduction of the
source term propagates through the contextual rule.

\begin{theorem}[Progress]
  \label{thm:Progress}
  If $\vdash t:A$ then either $t$ is a value or there exists $r$ such that
  $t \rightarrow r$.
\end{theorem}

One proves a stronger statement for open terms: if $\Gamma\vdash t:A$, then $t$
is a value, reduces, or contains a free variable and does not rewrite. For the
new case $\Gamma, \Delta \vdash \letxn{x}{t}{s} : A$: if $t$ is a density
matrix $\rho^n$ the main rule fires; if $t$ reduces, the contextual rule
applies; otherwise $t$ has a free variable and so does the whole term.

\begin{theorem}[Strong Normalisation]
  \label{thm:strnorm}
  Every closed, well-typed term $\vdash t : A$ is strongly normalising.
\end{theorem}

The proof uses Girard's reducibility candidates~\cite{GirardLafontTaylor89}.
Following Díaz-Caro and Dowek~\cite{DiazCaroDowek24}, we extend the reduction
with projection rules $\sum_i p_i\, t_i \rightarrow t_j$ to handle the
distributivity rule without duplication, and define typed reducibility sets
$\SN(A)$ with $\SN(n) = \SN((m,n)) = \SN$ and $\SN(A \multimap B) = \{ t \in \SN
\mid t \rightarrow^{*} \lambda x.\,u \Rightarrow \forall v \in \SN(A),\,
u[v/x] \in \SN(B) \}$. The key new per-construct lemma states that if $t \in
\SN(n)$ and $s[\rho_1/x_1,\ldots,\rho_n/x_n] \in \SN(A)$ for all single-qubit
$\rho_k \in \SN(1)$, then $\letxn{x}{t}{s} \in \SN(A)$. A fundamental
substitution lemma concludes strong normalisation for every closed well-typed
term.

\begin{remark}
  One may wonder whether Theorem~\ref{thm:strnorm} could be obtained by a
  simpler argument, showing that some measure on terms decreases along
  reduction: after all, the calculus is affine, so no variable is ever
  duplicated by substitution. Such an argument may well be possible, but it is
  not immediately obvious, because the terms of a linear combination can be
  repeated. The main $\mathsf{let}$ rule replaces its body $s$ by one copy of
  $s$ per summand of the decomposition, and both the $\mathsf{letcase}^\circ$
  rule and the distributivity rule
  $(\sum_i p_i t_i)r \rightarrow \sum_i p_i (t_i r)$ replicate whole subterms
  across the summands. Along these reductions a term therefore typically grows
  rather than shrinks, and no obvious measure decreases.
\end{remark}

\subsection{Soundness and Adequacy}

\begin{theorem}[Soundness]
  \label{thm:Soundness}
  If $\Gamma \vdash t : A$, $\theta \vDash \Gamma$, and $t \rightarrow r$, then
  $\fsem{t} = \fsem{r}$.
\end{theorem}

The proof is by induction on the derivation of $t\rightarrow r$.
The key insight behind the $\mathsf{let}$ case is that the operational semantics
(the reduction rule for $\mathsf{let}$) and the denotational semantics (the
interpretation of $\mathsf{let}$) are defined from the same decomposition given
by Proposition~\ref{prop:combined}.
Soundness then reduces to an application of the semantic substitution lemma
($\fsem{t[s/x]} = \fsem[\theta, x \mapsto \fsem{s}]{t}$) together with
linearity of the interpretation.

Soundness tells us the interpretation is invariant under reduction, but
invariance alone does not imply it carries useful information. Full completeness
(denotationally equal closed terms reduce to a common value) is too strong for
function types: the terms $r = \lambda f.\, f$ and $t = \lambda f.\, \lambda x.\,
f(x)$ are distinct normal forms with the same denotation (the identity CPM on
$A \multimap A$) and no common reduct. We therefore characterise semantic
equality by \emph{observational equivalence}.

\begin{definition}[Context, observational equivalence]
  \label{def:obs-equiv}
  A \emph{context} $C$ is a term with a single hole $[\cdot]$ such that
  $[\cdot]:A \vdash C : n$ for some type $A$ and $n \in \mathbb{N}$. Two closed
  terms $\vdash t : A$ and $\vdash r : A$ are \emph{observationally equivalent},
  written $t \equiv r$, if for every context $C$ with hole type $A$ there exists
  a density matrix $\rho$ such that $C[t] \rightarrow^{*} \rho$ and
  $C[r] \rightarrow^{*} \rho$.
\end{definition}

\begin{theorem}[Adequacy]
  \label{thm:Adequacy}
  If $\vdash t : A$, $\vdash r : A$, and $\fsem{t} = \fsem{r}$, then
  $t \equiv r$.
\end{theorem}

\begin{proof}
  (\emph{Sketch.}) Let $C$ be a context with hole type $A$ and output type $n$.
  By compositionality of the interpretation, $\fsem{t} = \fsem{r}$ implies
  $\fsem{C[t]} = \fsem{C[r]}$. By Theorems~\ref{thm:SR}, \ref{thm:Progress}, and~\ref{thm:strnorm},
  the closed ground-type terms $C[t]$ and $C[r]$ reduce to density
  matrices $\rho_1$ and $\rho_2$. By Soundness, $\rho_1 = \fsem{C[t]} =
  \fsem{C[r]} = \rho_2$. Hence both reduce to a common $\rho$, so
  $t \equiv r$. \qed
\end{proof}

Adequacy shows that semantic equality implies observational equivalence:
no ground-type experiment can distinguish two denotationally equal programs.

\section{Examples}
\label{sec:examples}

\subsection{Teleportation}
\label{sec:telep}

In $\qlambdens$, quantum teleportation can be written as the single
term~\cite{DiazCaro17}
\[
  T = \lambda x.\, \qletcase y {\pi^2(\had^1(\cnot^2 (x\otimes\beta_{00})))} {
    y,\, \Z_3 y,\, \X_3 y,\, \Z_3\X_3 y },
\]
where $\Z_3 = I \otimes I \otimes \Z$ and $\X_3 = I \otimes I \otimes \X$ operate
on the global 3-qubit state. This formulation conflates Alice's and Bob's roles:
Alice prepares and measures, while Bob corrects, yet the term treats the entire
protocol as a monolithic operation. Moreover, the type of $T$ is $1 \multimap 3$,
since the branches return the full 3-qubit post-measurement state, even though
the teleported qubit is only the third one.
This reflects an artefact of the encoding rather than the intended behaviour of
the protocol.

Using the $\mathsf{let}$ construct, we can cleanly separate the two roles. The
key observation is that since the type system is affine, Bob can bind all three
qubit variables via $\mathsf{let}$ but only use $q_3$, his qubit, discarding
Alice's qubits $q_1$ and $q_2$.
Alice takes the input qubit, composes it with her
half of the Bell pair, applies the CNOT and Hadamard gates, and measures her two
qubits. Bob then analyses the resulting measurement term and, in each branch,
decomposes the corresponding 3-qubit post-measurement state, applies the
appropriate Pauli correction to his qubit, and returns only that qubit:
\begin{align*}
  \mathit{alice} &= \lambda x.\, \pi^2(\had^1(\cnot^2(x \otimes \beta_{00}))), \\
  \mathit{bob} &= \lambda m.\, \mathsf{letcase}^\circ\ y = m\ \mathsf{in}\ \{
    \letxn{q}{y}{q_3},\
    \letxn{q}{y}{\Z^1 q_3}, \\
  &\hphantom{{}= \lambda m.\, \mathsf{letcase}^\circ\ y = m\ \mathsf{in}\ \{}
    \letxn{q}{y}{\X^1 q_3},\
    \letxn{q}{y}{\Z^1(\X^1 q_3)}
  \}, \\
  \mathit{teleport} &= \lambda x.\, \mathit{bob}(\mathit{alice}\ x).
\end{align*}
The correction $\Z^1 q_3$ applies $\Z$ to a \emph{single-qubit} variable, rather
than the global operator $\Z_3 = I \otimes I \otimes \Z$.
The type of $\mathit{teleport}$ is $1 \multimap 1$: it takes a single-qubit
density matrix and returns a single-qubit density matrix. This type accurately
reflects the physical protocol, where Alice's qubits are discarded after the
correction step and only Bob's corrected qubit remains.

To illustrate the role of the new typing rule, the derivation of a typical Bob
branch is
\[
  \infer[\mathsf{let}]{y\!:\!3 \vdash \letxn{q}{y}{\Z^1 q_3} : 1}{
    \infer[\mathsf{ax}]{y\!:\!3 \vdash y : 3}{} &
    \infer[\mathsf{u}]{q_1\!:\!1,q_2\!:\!1,q_3\!:\!1 \vdash \Z^1 q_3 : 1}{
      \infer[\mathsf{ax}]{q_1\!:\!1,q_2\!:\!1,q_3\!:\!1 \vdash q_3 : 1}{}
    }
  }
\]
where $q_1$ and $q_2$ appear in the context but are unused: this is precisely how
Alice's qubits are discarded, and by Proposition~\ref{prop:discard} this discard
computes the partial trace over them.

\subsection{Bit-Flip Error Correction}
\label{sec:bitflip}

The three-qubit bit-flip code~\cite[\S10.1.1]{NielsenChuang10} protects a single
logical qubit against a single bit-flip ($X$) error by encoding it into three
physical qubits, using two ancilla qubits during syndrome extraction.
This example illustrates how the $\mathsf{let}$ construct enables clean ancilla
management, qubit extraction, and a more faithful typing of the protocol.

In the original calculus, each branch of the correction must apply the
correction as a \emph{global} five-qubit operator and the ancilla qubits cannot
be discarded, so the correction has type $3 \multimap 5$ and, after decoding,
the full protocol has type $\mathit{bitflip}_0 : 1 \multimap 5$: there is no
mechanism to discard the ancilla qubits or extract the logical qubit, and the
ancillas are carried as irremovable overhead. With the $\mathsf{let}$ construct,
each branch decomposes the five-qubit post-measurement state, applies the
correction to only the affected code qubit, and recombines only the three code
qubits, discarding the ancillas; decoding then extracts the first qubit.
Writing $\beta = \ketbra{0}{0}$ for the ancilla state and $\mathrm{Enc}^3$,
$\mathrm{Synd}^5$, $\mathrm{Dec}^3 = (\mathrm{Enc}^3)^\dagger$ for the encoding,
syndrome-extraction and decoding circuits, after syndrome extraction the first
two qubits are ancillas carrying the syndrome information, while the last three
qubits contain the encoded state. Accordingly, the correction step discards the
ancillas and acts only on the three code qubits:
\begin{align*}
  \mathit{encode} &= \lambda x.\, \mathrm{Enc}^3(x \otimes \beta \otimes \beta) \;:\; 1 \multimap 3, \\
  \mathit{correct} &= \lambda y.\,
  \mathsf{letcase}^\circ\ z = \pi^2(\mathrm{Synd}^5((\beta \otimes \beta) \otimes y))\ \mathsf{in}\ \{b_0, b_1, b_2, b_3\}, \\
  b_1 &= \letxk{q}{5}{z}{(X^1\, q_3) \otimes q_4 \otimes q_5}, \quad\ldots\quad : 3 \text{ (with $z:5$)}, \\
  \mathit{decode} &= \lambda w.\, \letxk{q}{3}{\mathrm{Dec}^3\, w}{q_1} \;:\; 3 \multimap 1, \\
  \mathit{bitflip} &= \lambda x.\, \mathit{decode}(\mathit{correct}(\mathit{encode}(x))) \;:\; 1 \multimap 1.
\end{align*}
The corrections are manifestly local: $b_1$ applies $X^1$ to a single-qubit
variable $q_3$, rather than the global operator $X_3^5 = I \otimes I \otimes X
\otimes I \otimes I$. The type $1 \multimap 1$ accurately reflects the physical
reality of the protocol: a single qubit is encoded, protected, and recovered.
Contrast this with $\mathit{bitflip}_0 : 1 \multimap 5$ in the original calculus.

\section{Conclusion}
\label{sec:conclusion}

We have shown that a multi-qubit state embedded inside a term can, after all,
be split into individual qubits: the Pauli decomposition together with the
spectral decomposition of the Pauli matrices expresses any $n$-qubit density
matrix as a real linear combination of tensor products of single-qubit states.
Exploiting this, we extended $\qlambdens$ with a compositional $\mathsf{let}$
construct, gave it a rewrite system, an affine type system, and a denotational
semantics based on completely positive maps, and proved Subject Reduction,
Progress, Strong Normalisation, Soundness, and Adequacy. The construct is
physically principled: discarding an unused qubit computes exactly the partial
trace, in accordance with the no-deleting theorem. The teleportation and
bit-flip examples show that programs over states-as-constants can be written
compositionally, with types that reflect the underlying physical protocol. This
supports the broader thesis that keeping the quantum state as a constant of the
language---rather than in an external register---is a viable design, well
suited to reasoning about programs.

\begin{credits}
\subsubsection{\ackname}
This work was partially funded by the European Union through the Horizon-Europe MSCA-SE project QCOMICAL (101182520),
the PEPR integrated project EPiQ (ANR-22-PETQ-0007),
the ANII Fondo Clemente Estable project FCE-1-2025-1-186751, and the STIC-AmSud project QUASAR.

\subsubsection{\discintname}
The authors have no competing interests to declare that are relevant to the
content of this article.
\end{credits}

\bibliographystyle{splncs04}
\bibliography{biblio}

\begin{thebibliography}{10}
\providecommand{\url}[1]{\texttt{#1}}
\providecommand{\urlprefix}{URL }
\providecommand{\doi}[1]{https://doi.org/#1}

\bibitem{AvanziniBartheDalLagoICFP21}
Avanzini, M., Barthe, G., Dal~Lago, U.: On continuation-passing transformations
  and expected cost analysis. In: Proc. ACM Program. Lang. vol.~5, pp. 1--30
  (2021)

\bibitem{AvanziniDiazcaroHainryPechoux25}
Avanzini, M., D\'{\i}az-Caro, A., Hainry, E., P\'echoux, R.: Expectation-based
  analysis of higher-order quantum programs. Draft at {\tt arXiv:2504.18441}
  (2025)

\bibitem{Borgna19}
Borgna, A.: Simulaci{\'o}n del lambda c{\'a}lculo de matrices de densidad en el
  lambda c{\'a}lculo cu{\'a}ntico de {Selinger} y {Valiron}. Master's thesis,
  Universidad de Buenos Aires (2019)

\bibitem{Delbecque08}
Delbecque, Y.: Game semantics for quantum data. In: Quantum Physics and Logic /
  Developments in Computational Models (QPL/DCM@ICALP 2008). pp. 41--57 (2008)

\bibitem{DiazCaro17}
D{\'\i}az-Caro, A.: A lambda calculus for density matrices with classical and
  probabilistic controls. In: Programming Languages and Systems (APLAS 2017).
  LNCS, vol. 10695, pp. 448--467. Springer (2017), arXiv:1705.00097

\bibitem{DiazCaroDowek24}
D{\'\i}az-Caro, A., Dowek, G.: A linear linear lambda-calculus. Mathematical
  Structures in Computer Science  \textbf{34}(10),  1103--1137 (2024)

\bibitem{GirardLafontTaylor89}
Girard, J.Y., Lafont, Y., Taylor, P.: Proofs and Types. No.~7 in Cambridge
  Tracts in Theoretical Computer Science, Cambridge University Press (1989)

\bibitem{HasuoHoshino17}
Hasuo, I., Hoshino, N.: Semantics of higher-order quantum computation via
  geometry of interaction. Annals of Pure and Applied Logic  \textbf{168}(2),
  404--469 (2017)

\bibitem{KoskaBaboulinGazda24}
Koska, O., Baboulin, M., Gazda, A.: A tree-approach {Pauli} decomposition
  algorithm with application to quantum computing. In: IEEE International
  Conference on Quantum Computing and Engineering (QCE) (2024),
  arXiv:2403.11644

\bibitem{Miguez26}
Miguez, T.: A compositional extension of $\lambda_\rho^\circ$ via Pauli
  decomposition. Master's thesis, Universidad de Buenos Aires (2026)

\bibitem{NielsenChuang10}
Nielsen, M.A., Chuang, I.L.: Quantum Computation and Quantum Information.
  Cambridge University Press, 10th anniversary edition edn. (2010)

\bibitem{PatiBraunstein00}
Pati, A.K., Braunstein, S.L.: Impossibility of deleting an unknown quantum
  state. Nature  \textbf{404}(6774),  164--165 (2000)

\bibitem{SelingerValironMSCS06}
Selinger, P., Valiron, B.: A lambda calculus for quantum computation with
  classical control. Mathematical Structures in Computer Science
  \textbf{16}(3),  527--552 (2006)

\bibitem{WoottersZurekN82}
Wootters, W., Zurek, W.: A single quantum cannot be cloned. Nature
  \textbf{299}(5886),  802--803 (1982)

\end{thebibliography}

\newpage

\appendix

\section{Spectral Decomposition of the Pauli Matrices}
\label{app:pauli-spectral}

We compute the spectral decomposition of each of the four Pauli matrices
explicitly. In each case, we find the eigenvalues and an orthonormal basis of
eigenvectors, and verify that $M = \sum_i \lambda_i \ketbra{v_i}{v_i}$.

\paragraph{Identity $I$.}
The identity matrix $I = \begin{pmatrix} 1 & 0 \\ 0 & 1 \end{pmatrix}$ has
eigenvalue $\lambda = 1$ with multiplicity $2$. The standard basis vectors
$\ket{0}$ and $\ket{1}$ are eigenvectors, so:
\[
  I = 1 \cdot \ketbra{0}{0} + 1 \cdot \ketbra{1}{1}.
\]

\paragraph{Pauli $Z$.}
The matrix $Z = \begin{pmatrix} 1 & 0 \\ 0 & -1 \end{pmatrix}$ is already
diagonal, with eigenvalue $\lambda_1 = +1$ for eigenvector $\ket{0}$ and
$\lambda_2 = -1$ for eigenvector $\ket{1}$. The characteristic polynomial is
$\det(Z - \lambda I) = (1-\lambda)(-1-\lambda) = 0$, giving $\lambda = \pm 1$.
The spectral decomposition is:
\[
  Z = (+1)\ketbra{0}{0} + (-1)\ketbra{1}{1}
  = \begin{pmatrix} 1 & 0 \\ 0 & 0 \end{pmatrix} - \begin{pmatrix} 0 & 0 \\ 0 & 1 \end{pmatrix}
  = \begin{pmatrix} 1 & 0 \\ 0 & -1 \end{pmatrix} = Z. \quad\checkmark
\]

\paragraph{Pauli $X$.}
The matrix $X = \begin{pmatrix} 0 & 1 \\ 1 & 0 \end{pmatrix}$ has characteristic
polynomial $\det(X - \lambda I) = \lambda^2 - 1 = 0$, giving eigenvalues
$\lambda = \pm 1$. For $\lambda_1 = +1$, solving $(X - I)v = 0$ gives $v_1 =
v_2$, so the normalized eigenvector is $\ket{+} =
\frac{1}{\sqrt{2}}\begin{pmatrix} 1 \\ 1 \end{pmatrix}$. For $\lambda_2 = -1$,
solving $(X + I)v = 0$ gives $v_1 = -v_2$, so the normalized eigenvector is
$\ket{-} = \frac{1}{\sqrt{2}}\begin{pmatrix} 1 \\ -1 \end{pmatrix}$.
\begin{align*}
  (+1)\ketbra{+}{+} + (-1)\ketbra{-}{-}
  &= \frac{1}{2}\begin{pmatrix} 1 & 1 \\ 1 & 1 \end{pmatrix}
    - \frac{1}{2}\begin{pmatrix} 1 & -1 \\ -1 & 1 \end{pmatrix}
  = \frac{1}{2}\begin{pmatrix} 0 & 2 \\ 2 & 0 \end{pmatrix}
  = \begin{pmatrix} 0 & 1 \\ 1 & 0 \end{pmatrix} = X. \quad\checkmark
\end{align*}

\paragraph{Pauli $Y$.}
The matrix $Y = \begin{pmatrix} 0 & -i \\ i & 0 \end{pmatrix}$ has characteristic
polynomial $\det(Y - \lambda I) = \lambda^2 - 1 = 0$, giving eigenvalues
$\lambda = \pm 1$. For $\lambda_1 = +1$, solving $(Y - I)v = 0$ gives $v_1 =
-iv_2$; absorbing the global phase, the normalized eigenvector is conventionally
written $\ket{i} = \frac{1}{\sqrt{2}}(\ket{0} + i\ket{1})$, so $\ketbra{i}{i} =
\frac{1}{2}\begin{pmatrix} 1 & -i \\ i & 1 \end{pmatrix}$. For $\lambda_2 = -1$,
solving $(Y + I)v = 0$ gives $v_1 = iv_2$, and with $\ket{-i} =
\frac{1}{\sqrt{2}}(\ket{0} - i\ket{1})$ we have $\ketbra{-i}{-i} =
\frac{1}{2}\begin{pmatrix} 1 & i \\ -i & 1 \end{pmatrix}$.
\begin{align*}
  (+1)\ketbra{i}{i} + (-1)\ketbra{-i}{-i}
  &= \frac{1}{2}\begin{pmatrix} 1 & -i \\ i & 1 \end{pmatrix}
    - \frac{1}{2}\begin{pmatrix} 1 & i \\ -i & 1 \end{pmatrix} \\
  &= \frac{1}{2}\begin{pmatrix} 0 & -2i \\ 2i & 0 \end{pmatrix}
  = \begin{pmatrix} 0 & -i \\ i & 0 \end{pmatrix} = Y. \quad\checkmark
\end{align*}

\section{Bell State Computations}
\label{app:bell-computations}

This appendix contains the detailed computations for the Bell state examples in
Examples~\ref{ex:bell-pauli},~\ref{ex:bell-combined}, and~\ref{ex:bell-discard}.

\subsection{Pauli Decomposition Coefficients (Example~\ref{ex:bell-pauli})}
\label{app:bell-pauli-coeffs}

We compute $\bra{\beta_{00}}(M_1 \otimes M_2)\ket{\beta_{00}}$ for all
sixteen pairs $M_1, M_2 \in \{I, X, Y, Z\}$, using the Pauli actions:
\[
  \begin{array}{llll}
    I\ket{0} = \ket{0},  & X\ket{0} = \ket{1},
    & Y\ket{0} = i\ket{1},  & Z\ket{0} = \ket{0}, \\[3pt]
    I\ket{1} = \ket{1},  & X\ket{1} = \ket{0},
    & Y\ket{1} = -i\ket{0}, & Z\ket{1} = -\ket{1}.
  \end{array}
\]
In each case we first compute
$(M_1 \otimes M_2)\ket{\beta_{00}} =
\frac{1}{\sqrt{2}}(M_1\ket{0} \otimes M_2\ket{0}
+ M_1\ket{1} \otimes M_2\ket{1})$
and then take the inner product with
$\bra{\beta_{00}} = \frac{1}{\sqrt{2}}(\bra{00} + \bra{11})$.

\paragraph{Pairs with $M_1 = I$.}
\begin{align*}
  (I \otimes I)\ket{\beta_{00}}
  &= \tfrac{1}{\sqrt{2}}(\ket{00} + \ket{11}),
  &
  \bra{\beta_{00}}(I \otimes I)\ket{\beta_{00}}
  &= \tfrac{1}{2}(1 + 1) = 1.
  \\
  (I \otimes X)\ket{\beta_{00}}
  &= \tfrac{1}{\sqrt{2}}(\ket{01} + \ket{10}),
  &
  \bra{\beta_{00}}(I \otimes X)\ket{\beta_{00}}
  &= \tfrac{1}{2}(0 + 0) = 0.
  \\
  (I \otimes Y)\ket{\beta_{00}}
  &= \tfrac{1}{\sqrt{2}}(i\ket{01} - i\ket{10}),
  &
  \bra{\beta_{00}}(I \otimes Y)\ket{\beta_{00}}
  &= \tfrac{1}{2}(0 + 0) = 0.
  \\
  (I \otimes Z)\ket{\beta_{00}}
  &= \tfrac{1}{\sqrt{2}}(\ket{00} - \ket{11}),
  &
  \bra{\beta_{00}}(I \otimes Z)\ket{\beta_{00}}
  &= \tfrac{1}{2}(1 - 1) = 0.
\end{align*}

\paragraph{Pairs with $M_1 = X$.}
\begin{align*}
  (X \otimes I)\ket{\beta_{00}}
  &= \tfrac{1}{\sqrt{2}}(\ket{10} + \ket{01}),
  &
  \bra{\beta_{00}}(X \otimes I)\ket{\beta_{00}}
  &= \tfrac{1}{2}(0 + 0) = 0.
  \\
  (X \otimes X)\ket{\beta_{00}}
  &= \tfrac{1}{\sqrt{2}}(\ket{00} + \ket{11}),
  &
  \bra{\beta_{00}}(X \otimes X)\ket{\beta_{00}}
  &= \tfrac{1}{2}(1 + 1) = 1.
  \\
  (X \otimes Y)\ket{\beta_{00}}
  &= \tfrac{1}{\sqrt{2}}(i\ket{11} - i\ket{00}),
  &
  \bra{\beta_{00}}(X \otimes Y)\ket{\beta_{00}}
  &= \tfrac{1}{2}(-i + i) = 0.
  \\
  (X \otimes Z)\ket{\beta_{00}}
  &= \tfrac{1}{\sqrt{2}}(\ket{10} - \ket{01}),
  &
  \bra{\beta_{00}}(X \otimes Z)\ket{\beta_{00}}
  &= \tfrac{1}{2}(0 + 0) = 0.
\end{align*}

\paragraph{Pairs with $M_1 = Y$.}
\begin{align*}
  (Y \otimes I)\ket{\beta_{00}}
  &= \tfrac{1}{\sqrt{2}}(i\ket{10} - i\ket{01}),
  &
  \bra{\beta_{00}}(Y \otimes I)\ket{\beta_{00}}
  &= \tfrac{1}{2}(0 + 0) = 0.
  \\
  (Y \otimes X)\ket{\beta_{00}}
  &= \tfrac{1}{\sqrt{2}}(i\ket{11} - i\ket{00}),
  &
  \bra{\beta_{00}}(Y \otimes X)\ket{\beta_{00}}
  &= \tfrac{1}{2}(-i + i) = 0.
  \\
  (Y \otimes Y)\ket{\beta_{00}}
  &= \tfrac{1}{\sqrt{2}}(-\ket{11} - \ket{00}),
  &
  \bra{\beta_{00}}(Y \otimes Y)\ket{\beta_{00}}
  &= \tfrac{1}{2}(-1 - 1) = -1.
  \\
  (Y \otimes Z)\ket{\beta_{00}}
  &= \tfrac{1}{\sqrt{2}}(i\ket{10} + i\ket{01}),
  &
  \bra{\beta_{00}}(Y \otimes Z)\ket{\beta_{00}}
  &= \tfrac{1}{2}(0 + 0) = 0.
\end{align*}

\paragraph{Pairs with $M_1 = Z$.}
\begin{align*}
  (Z \otimes I)\ket{\beta_{00}}
  &= \tfrac{1}{\sqrt{2}}(\ket{00} - \ket{11}),
  &
  \bra{\beta_{00}}(Z \otimes I)\ket{\beta_{00}}
  &= \tfrac{1}{2}(1 - 1) = 0.
  \\
  (Z \otimes X)\ket{\beta_{00}}
  &= \tfrac{1}{\sqrt{2}}(\ket{01} - \ket{10}),
  &
  \bra{\beta_{00}}(Z \otimes X)\ket{\beta_{00}}
  &= \tfrac{1}{2}(0 + 0) = 0.
  \\
  (Z \otimes Y)\ket{\beta_{00}}
  &= \tfrac{1}{\sqrt{2}}(i\ket{01} + i\ket{10}),
  &
  \bra{\beta_{00}}(Z \otimes Y)\ket{\beta_{00}}
  &= \tfrac{1}{2}(0 + 0) = 0.
  \\
  (Z \otimes Z)\ket{\beta_{00}}
  &= \tfrac{1}{\sqrt{2}}(\ket{00} + \ket{11}),
  &
  \bra{\beta_{00}}(Z \otimes Z)\ket{\beta_{00}}
  &= \tfrac{1}{2}(1 + 1) = 1.
\end{align*}
This yields $\alpha_{II} = \alpha_{XX} = \alpha_{ZZ} = \frac14$,
$\alpha_{YY} = -\frac14$, and all others zero.

\subsection{Combined Decomposition Terms (Example~\ref{ex:bell-combined})}
\label{app:bell-combined-terms}

We expand each non-zero Pauli pair from equation~(\ref{eq:bell-pauli}) via the
spectral decomposition. For each pair, $p_{i\vec{l}} = \alpha_i \cdot
\lambda_{i_1}^{l_1} \cdot \lambda_{i_2}^{l_2}$.

\paragraph{From $I \otimes I$ ($\alpha_{II} = \frac{1}{4}$).}
Since all eigenvalues of $I$ are $+1$, every $p_{i\vec{l}} = \frac{1}{4}$:
\begin{align*}
  \vec{l} = (1,1) &\colon\; \tfrac{1}{4}\;\ketbra{0}{0} \otimes \ketbra{0}{0}
  & \vec{l} = (1,2) &\colon\; \tfrac{1}{4}\;\ketbra{0}{0} \otimes \ketbra{1}{1} \\
  \vec{l} = (2,1) &\colon\; \tfrac{1}{4}\;\ketbra{1}{1} \otimes \ketbra{0}{0}
  & \vec{l} = (2,2) &\colon\; \tfrac{1}{4}\;\ketbra{1}{1} \otimes \ketbra{1}{1}
\end{align*}

\paragraph{From $X \otimes X$ ($\alpha_{XX} = \frac{1}{4}$).}
The eigenvalues of $X$ are $+1$ and $-1$:
\begin{align*}
  \vec{l} = (1,1) &\colon\; \tfrac{1}{4}(+1)(+1)\;\ketbra{+}{+} \otimes \ketbra{+}{+}
    = \tfrac{1}{4}\;\ketbra{+}{+} \otimes \ketbra{+}{+} \\
  \vec{l} = (1,2) &\colon\; \tfrac{1}{4}(+1)(-1)\;\ketbra{+}{+} \otimes \ketbra{-}{-}
    = -\tfrac{1}{4}\;\ketbra{+}{+} \otimes \ketbra{-}{-} \\
  \vec{l} = (2,1) &\colon\; \tfrac{1}{4}(-1)(+1)\;\ketbra{-}{-} \otimes \ketbra{+}{+}
    = -\tfrac{1}{4}\;\ketbra{-}{-} \otimes \ketbra{+}{+} \\
  \vec{l} = (2,2) &\colon\; \tfrac{1}{4}(-1)(-1)\;\ketbra{-}{-} \otimes \ketbra{-}{-}
    = \tfrac{1}{4}\;\ketbra{-}{-} \otimes \ketbra{-}{-}
\end{align*}

\paragraph{From $Y \otimes Y$ ($\alpha_{YY} = -\frac{1}{4}$).}
The eigenvalues of $Y$ are $+1$ and $-1$:
\begin{align*}
  \vec{l} = (1,1) &\colon\; (-\tfrac{1}{4})(+1)(+1)\;\ketbra{i}{i} \otimes \ketbra{i}{i}
    = -\tfrac{1}{4}\;\ketbra{i}{i} \otimes \ketbra{i}{i} \\
  \vec{l} = (1,2) &\colon\; (-\tfrac{1}{4})(+1)(-1)\;\ketbra{i}{i} \otimes \ketbra{-i}{-i}
    = \tfrac{1}{4}\;\ketbra{i}{i} \otimes \ketbra{-i}{-i} \\
  \vec{l} = (2,1) &\colon\; (-\tfrac{1}{4})(-1)(+1)\;\ketbra{-i}{-i} \otimes \ketbra{i}{i}
    = \tfrac{1}{4}\;\ketbra{-i}{-i} \otimes \ketbra{i}{i} \\
  \vec{l} = (2,2) &\colon\; (-\tfrac{1}{4})(-1)(-1)\;\ketbra{-i}{-i} \otimes \ketbra{-i}{-i}
    = -\tfrac{1}{4}\;\ketbra{-i}{-i} \otimes \ketbra{-i}{-i}
\end{align*}

\paragraph{From $Z \otimes Z$ ($\alpha_{ZZ} = \frac{1}{4}$).}
The eigenvalues of $Z$ are $+1$ and $-1$:
\begin{align*}
  \vec{l} = (1,1) &\colon\; \tfrac{1}{4}(+1)(+1)\;\ketbra{0}{0} \otimes \ketbra{0}{0}
    = \tfrac{1}{4}\;\ketbra{0}{0} \otimes \ketbra{0}{0} \\
  \vec{l} = (1,2) &\colon\; \tfrac{1}{4}(+1)(-1)\;\ketbra{0}{0} \otimes \ketbra{1}{1}
    = -\tfrac{1}{4}\;\ketbra{0}{0} \otimes \ketbra{1}{1} \\
  \vec{l} = (2,1) &\colon\; \tfrac{1}{4}(-1)(+1)\;\ketbra{1}{1} \otimes \ketbra{0}{0}
    = -\tfrac{1}{4}\;\ketbra{1}{1} \otimes \ketbra{0}{0} \\
  \vec{l} = (2,2) &\colon\; \tfrac{1}{4}(-1)(-1)\;\ketbra{1}{1} \otimes \ketbra{1}{1}
    = \tfrac{1}{4}\;\ketbra{1}{1} \otimes \ketbra{1}{1}
\end{align*}

\subsection{Reduction of the \texorpdfstring{$\mathsf{let}$}{let} Term (Example~\ref{ex:bell-discard})}
\label{app:bell-let-terms}

After applying the $\mathsf{let}$ reduction rule to
$\letxk{x}{2}{\beta_{00}}{x_2}$ and dropping the vacuous substitution
on $x_1$, the 16 terms (grouped by Pauli pair) are:
\begin{align}
  &
  \phantom{+\;}
    \tfrac{1}{4}\ketbra{0}{0}
  + \tfrac{1}{4}\ketbra{1}{1}
  + \tfrac{1}{4}\ketbra{0}{0}
  + \tfrac{1}{4}\ketbra{1}{1}
  \nonumber\\
  &+ \tfrac{1}{4}\ketbra{+}{+}
  - \tfrac{1}{4}\ketbra{-}{-}
  - \tfrac{1}{4}\ketbra{+}{+}
  + \tfrac{1}{4}\ketbra{-}{-}
  \nonumber\\
  &- \tfrac{1}{4}\ketbra{i}{i}
  + \tfrac{1}{4}\ketbra{-i}{-i}
  + \tfrac{1}{4}\ketbra{i}{i}
  - \tfrac{1}{4}\ketbra{-i}{-i}
  \nonumber\\
  &+ \tfrac{1}{4}\ketbra{0}{0}
  - \tfrac{1}{4}\ketbra{1}{1}
  - \tfrac{1}{4}\ketbra{0}{0}
  + \tfrac{1}{4}\ketbra{1}{1}
  \label{eq:bell-let-16terms}
\end{align}
Collecting coefficients of equal density matrices:
\begin{alignat*}{3}
  \ketbra{0}{0} &\colon\quad
  &&\tfrac{1}{4} + \tfrac{1}{4} + \tfrac{1}{4} - \tfrac{1}{4}
  &&= \tfrac{1}{2} \\
  \ketbra{1}{1} &\colon\quad
  &&\tfrac{1}{4} + \tfrac{1}{4} - \tfrac{1}{4} + \tfrac{1}{4}
  &&= \tfrac{1}{2} \\
  \ketbra{+}{+} &\colon\quad
  &&\tfrac{1}{4} - \tfrac{1}{4}
  &&= 0 \\
  \ketbra{-}{-} &\colon\quad
  &&\tfrac{1}{4} - \tfrac{1}{4}
  &&= 0 \\
  \ketbra{i}{i} &\colon\quad
  &&\tfrac{1}{4} - \tfrac{1}{4}
  &&= 0 \\
  \ketbra{-i}{-i} &\colon\quad
  &&\tfrac{1}{4} - \tfrac{1}{4}
  &&= 0
\end{alignat*}
so the result is $\tfrac{1}{2}\ketbra{0}{0} + \tfrac{1}{2}\ketbra{1}{1} =
\tfrac{I}{2} = \tr_1(\beta_{00})$.

\section{Detailed Proofs}
\label{app:proofs}

Throughout, we assume standard weakening and strengthening properties for the
affine type system.

\subsection{Substitution Lemma}

\begin{lemma}[Substitution]
  \label{lem:substitution}
  If $\Gamma,x:A\vdash t:B$ and $\Delta\vdash r:A$, then
  $\Gamma,\Delta\vdash t[r/x]:B$.
\end{lemma}

\begin{proof}
By induction on $t$.
\begin{itemize}
\item Let $t=x$. Then $B=A$. By weakening, $\Gamma,\Delta\vdash r:A$.
  Notice that $t[r/x]=r$.
\item Let $t=y$. Then, by weakening and strengthening,
  $\Gamma,\Delta\vdash y:B$. Notice that $t[r/x]=y$.
\item Let $t=\lambda y.s$. Then $B=C\multimap D$ and, by inversion,
  $\Gamma,x:A,y:C\vdash s:D$. Then, by the induction hypothesis,
  $\Gamma,y:C,\Delta\vdash s[r/x]:D$, so, by rule $\multimap_i$,
  $\Gamma,\Delta\vdash\lambda y.(s[r/x]):C\multimap D$. Notice that $\lambda
  y.(s[r/x])=(\lambda y.s)[r/x]$.
\item Let $t=t_1t_2$. Then $\Gamma,x:A=\Gamma_1,\Gamma_2$, with $\Gamma_1\vdash
  t_1:C\multimap B$ and $\Gamma_2\vdash t_2:C$.
  \begin{itemize}
  \item If $x:A\in\Gamma_1$, then, by the induction hypothesis
    $\Gamma_1\setminus\{x:A\},\Delta\vdash t_1[r/x]:C\multimap B$, so by rule
    $\multimap_e$, $\Gamma_1\setminus\{x:A\},\Gamma_2,\Delta\vdash
    t_1[r/x]t_2:B$. Notice that $\Gamma_1\setminus\{x:A\},\Gamma_2=\Gamma$ and
    $t_1[r/x]t_2=(t_1t_2)[r/x]$.
  \item If $x:A\in\Gamma_2$, then, by the induction hypothesis
    $\Gamma_2\setminus\{x:A\},\Delta\vdash t_2[r/x]:C$, so by rule
    $\multimap_e$, $\Gamma_1,\Gamma_2\setminus\{x:A\},\Delta\vdash
    t_1(t_2[r/x]):B$. Notice that $\Gamma_1,\Gamma_2\setminus\{x:A\}=\Gamma$ and
    $t_1(t_2[r/x])=(t_1t_2)[r/x]$.
  \end{itemize}
\item Let $t=\rho^n$. Then $B=n$. By weakening and strengthening,
  $\Gamma,\Delta\vdash\rho^n:n$. Notice that $t[r/x]=\rho^n$.
\item Let $t=U^ms$. Then $B=n$ and $\Gamma,x:A\vdash s:n$. Then, by the
  induction hypothesis, $\Gamma,\Delta\vdash s[r/x]:n$. So, by rule $u$,
  $\Gamma,\Delta\vdash U^m(s[r/x]):n$. Notice that $U^m(s[r/x])=(U^ms)[r/x]$.
\item Let $t=\pi^ms$. Then $B=(m,n)$ and $\Gamma,x:A\vdash s:n$. Then, by the
  induction hypothesis, $\Gamma,\Delta\vdash s[r/x]:n$. So, by rule $m$,
  $\Gamma,\Delta\vdash \pi^m(s[r/x]):(m,n)$. Notice that
  $\pi^m(s[r/x])=(\pi^ms)[r/x]$.
\item Let $t=t_1\otimes t_2$. Then $B=n_1+n_2$, $\Gamma,x:A=\Gamma_1,\Gamma_2$
  with $\Gamma_i\vdash t_i:n_i$ for $i=1,2$. Let $x:A\in\Gamma_i$ for some
  $i=1,2$. Then, by the induction hypothesis,
  $\Gamma_i\setminus\{x:A\},\Delta\vdash t_i[r/x]$, so by rule $\otimes$, either
  $\Gamma,\Delta\vdash t_1[r/x]\otimes t_2:n_1+n_2$, or $\Gamma,\Delta\vdash
  t_1\otimes t_2[r/x]:n_1+n_2$. In the first case, notice that $t_1[r/x]\otimes
  t_2=(t_1\otimes t_2)[r/x]$, and in the second, $t_1\otimes
  t_2[r/x]=(t_1\otimes t_2)[r/x]$.
\item Let $t=\sum_ip_it_i$. Then $\Gamma,x:A\vdash t_i:B$ and so, by the
  induction hypothesis, $\Gamma,\Delta\vdash t_i[r/x]:B$. Therefore, by rule
  $+$, $\Gamma,\Delta\vdash\sum_ip_it_i[r/x]:B$. Notice that $\sum_ip_it_i[r/x]
  = (\sum_ip_it_i)[r/x]$.
\item Let $t=\qletcase{y}{s}{t_0,\dots,t_{2^m-1}}$. $y:n\vdash t_i:B$, for
  $i=0,\dots,2^m-1$, and, $\Gamma\vdash s:(m,n)$. By the induction hypothesis,
  $\Gamma,\Delta\vdash s[r/x]:(m,n)$. So, by rule $\mathsf{lc}$,
  $\Gamma,\Delta\vdash\qletcase{y}{s[r/x]}{t_0,\dots,t_{2^m-1}}\!\!:B$. Notice
  that $(\qletcase{y}{s}{t_0,\dots,t_{2^m-1}})[r/x] = \qletcase{y}{s[r/x]}{t_0,\dots,t_{2^m-1}}$.
\item Let $t = \letxn{y}{s}{u}$. Then $\Gamma, x:A = \Gamma_1, \Gamma_2$ with
  $\Gamma_1 \vdash s : n$ and $\Gamma_2, y_1:1, \ldots, y_n:1 \vdash u : B$.
  \begin{itemize}
  \item If $x:A \in \Gamma_1$, then by the induction hypothesis,
    $\Gamma_1 \setminus \{x:A\}, \Delta \vdash s[r/x] : n$.
    Since $x \notin FV(u)$, we have $u[r/x] = u$.
    By rule $\mathsf{let}$,
    $\Gamma_1 \setminus \{x:A\}, \Gamma_2, \Delta \vdash \letxn{y}{s[r/x]}{u} : B$.
  \item If $x:A \in \Gamma_2$, then by the induction hypothesis,
    $\Gamma_2 \setminus \{x:A\}, \Delta, y_1:1, \ldots, y_n:1 \vdash u[r/x] : B$.
    Since $x \notin FV(s)$, we have $s[r/x] = s$.
    By rule $\mathsf{let}$,
    $\Gamma_1, \Gamma_2 \setminus \{x:A\}, \Delta \vdash \letxn{y}{s}{u[r/x]} : B$.
    \qed
  \end{itemize}
\end{itemize}
\end{proof}

\subsection{Subject Reduction (Theorem~\ref{thm:SR})}

\begin{proof}~
\begin{itemize}
\item Let $t=(\lambda x.t')s$, $r=t'[s/x]$. Then
  $\Gamma\vdash(\lambda x.t')s:A$, so, $\Gamma_1\vdash\lambda x.t':B\multimap
  A$ and $\Gamma_2\vdash s:B$, with $\Gamma=\Gamma_1,\Gamma_2$. Hence,
  $\Gamma_1,x:B\vdash t':A$, and so, by Lemma~\ref{lem:substitution},
  $\Gamma\vdash t'[s/x]:A$.
\item Let $t=U^m\rho^n$, $r={\rho'}^n$, with
  ${\rho'}^n=\ext{U^m}\rho^n\ext{U^m}^\dagger$. Then $A=n$. By rule
  $\mathsf{ax}_\rho$, $\Gamma\vdash{\rho'}^n:n$.
\item Let $t=\rho_1^n\otimes\rho_2^m$ and $r=\rho$, with
  $\rho=\rho_1^n\otimes\rho_2^m$. Then, $A=n+m$, with $\vdash\rho_1^n:n$ and
  $\vdash\rho_2^m:m$. Since $\rho$ is a density matrix of $(n+m)$-qubits,
  $\vdash\rho:n+m$.
\item Let $t=\qletcase{x}{\pi^m\rho^n}{t_0,\dots,t_{2^m-1}}$ and $r=\sum_i p_i
  t_i[\rho^n_i/x]$, with
  $\rho_i^n=\frac{\ext{\pi_i}\rho^n\ext{\pi_i}^\dagger}{p_i}$ and
  $p_i=\tr(\ext{\pi_i}^\dagger\ext{\pi_i}\rho^n)$. Then
  $\Gamma\vdash\pi^m\rho^n:(m,n)$ and $x:n\vdash t_i:A$. By
  Lemma~\ref{lem:substitution}, $\Gamma\vdash t_i[\rho_i^n/x]:A$, then, by
  rule $+$, $\Gamma\vdash\sum_ip_it_i[\rho_i^n/x]:A$.
\item Let $t=\qletcase{x}{\sum_j q_j \hat w_j}{s_0,\dots,s_{2^m-1}}$ and
  $r=\sum_j q_j\,\qletcase{x}{\hat w_j}{s_0,\dots,s_{2^m-1}}$.
  By inversion of rule $\mathsf{lc}$, $\Gamma\vdash\sum_j q_j \hat w_j:(m,n)$
  and $x:n\vdash s_i:A$ for each $i$.  By rule $+$, each $\hat w_j:(m,n)$.
  Applying rule $\mathsf{lc}$ to each $\hat w_j$ gives
  $\Gamma\vdash\qletcase{x}{\hat w_j}{s_0,\dots,s_{2^m-1}}:A$, and then rule $+$
  gives $\Gamma\vdash r:A$.
\item Let $t=\sum_i p_i\rho_i$ and $r=\rho'$, with $\rho'=\sum_ip_i\rho_i$.
  Then, $\Gamma\vdash\sum_ip_i\rho_i:n$, and by rule $\mathsf{ax}_\rho$,
  $\Gamma\vdash\rho':n$.
\item Let $t=\sum_i p_i r$. Then, $\Gamma\vdash r:A$.
\item Let $t=(\sum_i p_i t_i)r$ and $\sum_i p_i (t_ir)$. Then,
  $\Gamma=\Gamma_1,\Gamma_2$, $\Gamma_1\vdash t_i:B\multimap A$ and
  $\Gamma_2\vdash r:B$. Therefore, by rule $\multimap_e$,
  $\Gamma_1,\Gamma_2\vdash t_ir:A$, and by rule $+$,
  $\Gamma_1,\Gamma_2\vdash\sum_ip_i(t_ir):A$.
\item Let $t = \letxn{x}{\rho^n}{s}$ and $r = \sum_{i=1}^{4^n} \sum_{\vec{l}
  \in \{1,2\}^n} p_{i\vec{l}}\, s[\rho_{i_{1}}^{l_1}/x_1, \cdots,
  \rho_{i_{n}}^{l_n}/x_n]$, where $p_{i\vec{l}}$ are the combined
  decomposition coefficients. We have $\Gamma, \Delta \vdash \letxn{x}{\rho^n}{s} : A$. By
  inversion, $\Gamma \vdash \rho^n : n$ and $\Delta, x_1:1, \ldots, x_n:1
  \vdash s : A$.
  As each $\rho_{i_{k}}^{l_k}$ is a 1-qubit density matrix by construction,
  by rule $\mathsf{ax_\rho}$ we have $\vdash \rho_{i_{k}}^{l_k} : 1$.
  Then, by repeated application of Lemma~\ref{lem:substitution}, we have that
  $\Delta \vdash s[\rho_{i_{1}}^{l_1}/x_1, \cdots, \rho_{i_{n}}^{l_n}/x_n] : A$.
  Since the Pauli decomposition preserves trace, we have that
  $\sum_{i=1}^{4^n} \sum_{\vec{l} \in \{1,2\}^n} p_{i\vec{l}} = 1$. Then, by
  rule $+$, we have $\Gamma, \Delta \vdash \sum_{i=1}^{4^n} \sum_{\vec{l}
  \in \{1,2\}^n} p_{i\vec{l}}\, s[\rho_{i_{1}}^{l_1}/x_1, \cdots,
  \rho_{i_{n}}^{l_n}/x_n] : A$.
\item Contextual cases: Let $s \rightarrow s'$, then
  \begin{itemize}
  \item Consider $t=t's$ and $r=t's'$. Then $\Gamma=\Gamma_1,\Gamma_2$, with
    $\Gamma_1\vdash t':B\multimap A$ and $\Gamma_2\vdash s:B$. By the
    induction hypothesis, $\Gamma_2\vdash s':B$, so by rule $\multimap_e$,
    $\Gamma\vdash t's':A$.
  \item Consider $t=st'$ and $r=s't'$. Then $\Gamma=\Gamma_1,\Gamma_2$, with
    $\Gamma_1\vdash s:B\multimap A$ and $\Gamma_2\vdash t':B$. By the
    induction hypothesis, $\Gamma_1\vdash s':B\multimap A$, so by rule
    $\multimap_e$, $\Gamma\vdash s't':A$.
  \item Consider $t=U^ms$ and $r=U^ms'$. Then $A=n$ and $\Gamma\vdash s:n$. By
    the induction hypothesis $\Gamma\vdash s':n$, so by rule $\mathsf{u}$,
    $\Gamma\vdash U^ms':n$.
  \item Consider $t=\pi^ms$ and $r=\pi^ms'$. Then $A=(m, n)$ and $\Gamma\vdash
    s:n$. By the induction hypothesis $\Gamma\vdash s':n$, so by rule
    $\mathsf{m}$, $\Gamma\vdash \pi^ms':(m,n)$.
  \item Consider $t=t'\otimes s$ and $r=t'\otimes s'$. Then $A=n+m$ and
    $\Gamma=\Gamma_1,\Gamma_2$, with $\Gamma_1\vdash t':n$ and $\Gamma_2\vdash
    s:m$. By the induction hypothesis $\Gamma_2\vdash s':m$, so by rule
    $\otimes$, $\Gamma\vdash t'\otimes s':n+m$.
  \item Consider $t=s\otimes t'$ and $r=s'\otimes t'$. Then $A=n+m$ and
    $\Gamma=\Gamma_1,\Gamma_2$, with $\Gamma_1\vdash s:n$ and $\Gamma_2\vdash
    t':m$. By the induction hypothesis $\Gamma_1\vdash s':n$, so by rule
    $\otimes$, $\Gamma\vdash s'\otimes t':n+m$.
  \item Consider $t=\qletcase{x}{s}{t_0,\dots,t_{2^m-1}}$ and $r=\qletcase{x}{s'}{t_0,\dots,t_{2^m-1}}$. Then $x:n\vdash t_i:A$ for $i=0,\dots,2^m-1$,
    and $\Gamma\vdash s:(m,n)$. By the induction hypothesis, $\Gamma\vdash
    s':(m,n)$ and by rule $\mathsf{lc}$, $\Gamma\vdash\qletcase{x}{s'}{t_0,\dots,t_{2^m-1}}:A$.
  \item Consider $t=\sum_ip_it_i$ and $r=\sum_ip_ir_i$, with $t_j\rightarrow r_j$, and
    $\forall i\neq j$, $t_i=r_i$. By inversion, $\Gamma\vdash t_i:A$. By the induction
    hypothesis, $\forall i$, $\Gamma\vdash r_i:A$. Then, by rule $+$,
    $\Gamma\vdash\sum_ip_ir_i:A$.
  \item Consider $t=\letxn{y}{s}{u}$ and $r=\letxn{y}{s'}{u}$. Then
    $\Gamma=\Gamma_1,\Gamma_2$, with $\Gamma_1\vdash s:n$ and
    $\Gamma_2,y_1:1,\ldots,y_n:1\vdash u:A$. By the induction hypothesis
    $\Gamma_1\vdash s':n$, so by rule $\mathsf{let}$,
    $\Gamma\vdash \letxn{y}{s'}{u}:A$.
    \qed
  \end{itemize}
\end{itemize}
\end{proof}

\subsection{Progress (Theorem~\ref{thm:Progress})}

\begin{proof}
  We relax the hypotheses and prove the theorem for open terms as well. That is:
  if $\Gamma\vdash t:A$, then either $t$ is a value, there exists $r$
  such that $t\rightarrow r$, or $t$ contains a free variable, and $t$ does not
  rewrite.
  \begin{itemize}
  \item Let $\Gamma,x:A\vdash x:A$ as a consequence of rule $\mathsf{ax}$. Then,
    we are done since $x$ is a free variable and does not rewrite.
  \item Let $\Gamma\vdash\lambda x.t:A\multimap B$ as a consequence of
    $\Gamma,x:A\vdash t:B$ and rule $\multimap_i$. Since reduction under
    $\lambda$ is not part of the rewrite system, $\lambda x.t$ is always a
    value by Definition~\ref{def:values}.
  \item Let $\Gamma,\Delta\vdash tr:B$ as a consequence of $\Gamma\vdash
    t:A\multimap B$, $\Delta\vdash r:A$ and rule $\multimap_e$. Then, by the
    induction hypothesis, one of the following cases happens:
    \begin{itemize}
    \item There exists $t'$ such that $t\rightarrow t'$, in which case $tr\rightarrow t'r$.
    \item There exists $r'$ such that $r\rightarrow r'$, in which case $tr\rightarrow tr'$.
    \item $t$ is a value and $r$ does not rewrite. The only values which
      can be typed by $A\multimap B$ are:
      \begin{itemize}
      \item $t=x$, in which case $xr$ contains a free variable and does not
        rewrite.
      \item $t=\lambda x.v$, in which case $(\lambda x.v)r\rightarrow v[r/x]$.
      \item $t=\sum_i p_i t_i$, where $\Gamma\vdash t_i:A\multimap B$. Then,
        $tr\rightarrow\sum_i p_i (t_ir)$.
      \end{itemize}
    \item $t$ is not a value, contains a free variable, and does not rewrite,
      and $r$ does not rewrite, in which case, if $t$ is not a sum, $tr$
      contains a free variable and does not rewrite. If $t=\sum_ip_it_i$ is a
      sum, $tr\rightarrow\sum_ip_i(t_ir)$.
    \end{itemize}
  \item Let $\Gamma\vdash\rho^n:n$ as a consequence of rule $\mathsf{ax_\rho}$.
    Then, we are done since $\rho^n$ is a value.
  \item Let $\Gamma\vdash U^mt:n$ as a consequence of $\Gamma\vdash t:n$ and
    rule $\mathsf{u}$. Then, by the induction hypothesis, one of the following
    cases happens:
    \begin{itemize}
    \item $t$ is a value. Since $\lambda x.v$ cannot be typed by $n$, the only
      values that can be typed by $n$ are either $\rho^n$, or they contain
      free variables:
      \begin{itemize}
      \item Let $t=\rho^n$, then $U^m\rho^n\rightarrow\rho'$, with
        $\rho'=\ext{U^m}\rho^n\ext{U^m}^\dagger$.
      \item Let $t$ contain a free variable. Notice that it can only be either a
        free variable by itself, a tensor of values containing free
        variables, or a linear combination of different values containing
        free variables. In any case, $t$ contains a free variable and does not
        rewrite. Hence, $U^mt$ contains a free variable and does not rewrite.
      \end{itemize}
    \item There exists $r$ such that $t\rightarrow r$, in which case $U^m t \rightarrow U^m r$;
    \item $t$ contains a free variable and does not rewrite, in which case the
      same is true for $U^mt$.
    \end{itemize}
  \item Let $\Gamma\vdash \pi^mt:(m,n)$ as a consequence of $\Gamma\vdash t:n$
    and rule $\mathsf{m}$. Then, by the induction hypothesis, one of the
    following cases happens:
    \begin{itemize}
    \item $t$ is a value. Then $\pi^m t$ is also a value and does not rewrite.
    \item $t$ contains a free variable and does not rewrite, in which case the
      same is true for $\pi^mt$.
    \end{itemize}
  \item Let $\Gamma,\Delta,\vdash t\otimes r:n+m$ as a consequence of
    $\Gamma\vdash t:n$, $\Delta\vdash r:m$ and rule $\otimes$. Then, by the
    induction hypothesis, one of the following happens:
    \begin{itemize}
    \item There exists $t'$ such that $t\rightarrow t'$, in which case $t\otimes r\rightarrow
      t'\otimes r$.
    \item There exists $r'$ such that $r\rightarrow r'$, in which case $t\otimes r\rightarrow
      t\otimes r'$.
    \item $t$ is a value and $r$ does not rewrite. The only values that
      can be typed by $n$ are either $\rho^n$, or they contain free variables.
      \begin{itemize}
      \item Let $t=\rho^n$, then:
        \begin{itemize}
        \item If $r=\rho^m$, $t\otimes r\rightarrow \rho'$, with
          $\rho'=\rho^n\otimes\rho^m$.
        \item If $r$ contains a free variable and does not rewrite, then the
          same is true for $t\otimes r$.
        \end{itemize}
      \item Let $t$ contain a free variable. Then $t\otimes r$ contains a free
        variable and does not rewrite.
      \end{itemize}
    \item $t$ contains a free variable and does not rewrite, and $r$ does not
      rewrite, in which case $t\otimes r$ contains a free variable and does not
      rewrite.
    \end{itemize}
  \item Let $\Gamma\vdash\qletcase{x}{r}{t_0,\dots,t_{2^m-1}}:A$ as a consequence
    of $x:n\vdash t_i:A$ for $i=0,\dots,2^m-1$, $\Gamma\vdash r:(m,n)$, and rule
    $\mathsf{lc}$. By the induction hypothesis, the possible cases for $r$ are:
    \begin{itemize}
    \item $r=\pi^m\rho^n$ (a measurement value on a concrete density matrix).
      Then the main $\mathsf{letcase}^\circ$ rule fires:
      \[\qletcase{x}{\pi^m\rho^n}{t_0,\dots,t_{2^m-1}}\rightarrow \sum_ip_it_i[\rho^n_i/x].\]
    \item $r=\sum_j q_j w_j$ is a closed non-trivial linear combination of
      values $w_j:(m,n)$ (necessarily each $w_j=\pi^m\rho^n_j$).  Then the
      distributing rule fires:
      \[\begin{aligned}
        &\qletcase{x}{\sum_j q_j w_j}{t_0,\dots,t_{2^m-1}}\\
        &\qquad\rightarrow \sum_j q_j\,\qletcase{x}{w_j}{t_0,\dots,t_{2^m-1}}.
      \end{aligned}\]
    \item $r$ contains a free variable and does not rewrite, in which case the
      same is true for $\qletcase{x}{r}{t_0,\dots,t_{2^m-1}}$.
    \item There exists $r'$ such that $r\rightarrow r'$, in which case, by contextual
      closure,
      \[\qletcase{x}{r}{t_0,\dots,t_{2^m-1}}\rightarrow\qletcase{x}{r'}{t_0,\dots,t_{2^m-1}}.\]
    \end{itemize}
  \item Let $\Gamma\vdash\sum_ip_it_i:A$ as a consequence of $\Gamma\vdash
    t_i:A$, $\sum_ip_i=1$, and rule $+$. If $\sum_ip_it_i$ is a value, then we
    are done. If it is not a value, then one of the following cases is true:
    \begin{itemize}
    \item $t_j=t_k$ for some $j\neq k$, in which case $\sum_ip_it_i\rightarrow
      (\sum_{i\neq j,k}p_it_i)+(p_j+p_k)t_j$.
    \item At least one $t_i$ is not a value. By the induction hypothesis, if
      $t_i$ is not a value, either it rewrites, or it contains a free variable
      and does not rewrite. If at least one $t_i$ rewrites, then $\sum_ip_it_i$
      rewrites. If none of these rewrites and at least one contains a free
      variable, then $\sum_ip_it_i$ does not rewrite and contain a free
      variable.
    \end{itemize}
  \item Let $\Gamma, \Delta \vdash \letxn{x}{t}{s} : A$
    as a consequence of $\Gamma \vdash t : n$, $\Delta, x_1:1, \ldots, x_n:1 \vdash s : A$
    and rule $\mathsf{let}$. By induction, the possible cases for $t$ are:
    \begin{itemize}
    \item $t$ is a value. Since $\lambda x.v$ cannot be typed by $n$, the only
      values that can be typed by $n$ are either $\rho^n$, or they contain
      free variables:
      \begin{itemize}
      \item Let $t=\rho^n$, then
        \[\letxn{x}{\rho^n}{s} \rightarrow \sum_{i=1}^{4^n} \sum_{l_1, \cdots, l_n \in \{1,2\}} p_{i\vec{l}}\, s[\rho_{i_{1}}^{l_1}/x_1, \cdots, \rho_{i_{n}}^{l_n}/x_n].\]
      \item Let $t$ contain a free variable. Notice that it can only be either a
        free variable by itself, a tensor of values containing free
        variables, or a linear combination of different values containing
        free variables. In any case, $t$ contains a free variable and does not
        rewrite. Hence, $\letxn{x}{t}{s}$ contains a free variable and does not rewrite.
      \end{itemize}
    \item There exists $t'$ such that $t \rightarrow t'$, in which case $\letxn{x}{t}{s} \rightarrow \letxn{x}{t'}{s}$.
      \item $t$ contains a free variable and does not rewrite, in which case the
        same is true for $\letxn{x}{t}{s}$.
	\qed
    \end{itemize}
  \end{itemize}
\end{proof}

\subsection{Strong Normalisation (Theorem~\ref{thm:strnorm})}

The proof uses Girard's reducibility candidates method~\cite{GirardLafontTaylor89}:
strong normalisation is established by defining typed reducibility sets $\SN(A)$,
proving a forward-closure lemma (Lemma~\ref{lem:cr3}), establishing
per-construct closure lemmas, and concluding by a fundamental substitution
lemma.

\paragraph{Extended reduction relation.}
Following D\'iaz-Caro and Dowek~\cite{DiazCaroDowek24}, we augment the
reduction relation with the following \emph{projection rules}:
\[
  \textstyle\sum_i p_i\, t_i \;\rightarrow\; t_j
  \qquad \text{for any index } j.
\]
These rules are not part of the operational semantics of the calculus; they are
added solely for the purpose of the strong normalisation argument. Since the
original reduction relation is a subset of the extended one, strong normalisation
for the extended relation implies strong normalisation for the original.

\paragraph{The sets $\SN$ and $\SN(A)$.}
Let $\SN$ denote the set of all terms that are strongly normalising under the
extended relation. The \emph{length} $\ell(t)$ of a term $t \in \SN$ is the
maximal length of any reduction sequence from $t$.
We define, by induction on the type $A$, a set $\SN(A)$ of \emph{reducible}
terms of type $A$:
\begin{align*}
  \SN(n) &= \SN, \qquad \SN((m,n)) = \SN, \\
  \SN(A \multimap B) &= \bigl\{\, t \in \SN \;\big|\;
    t \rightarrow^{*} \lambda x.\,u \;\Rightarrow \\
  &\hphantom{{}={}} \forall v \in \SN(A),\;
    u[v/x] \in \SN(B) \,\bigr\}.
\end{align*}

\begin{lemma}[$\SN$ is closed under sums]
  \label{lem:sn-sum}
  If $t_1, \ldots, t_k \in \SN$ then $\sum_i p_i\, t_i \in \SN$.
\end{lemma}
\begin{proof}
  By induction on $\sum_i \ell(t_i)$ and then on the size of the term.
  Any one-step reduct of $\sum_i p_i\, t_i$ either reduces one summand
  (covered by the induction hypothesis on $\ell$), applies a sum-collapse rule
  to produce a density matrix or a single term $t$ (both in $\SN$), or fires a
  projection rule to produce some $t_j$ (in $\SN$ by assumption). In each case
  the reduct is in $\SN$, so $\sum_i p_i\, t_i \in \SN$.
  \qed
\end{proof}

\begin{definition}[Neutral terms]
  A term $t$ is \emph{neutral} if it is not a $\lambda$-abstraction and not a
  probabilistic sum.
\end{definition}

\begin{lemma}[CR3: forward closure for neutral terms]
  \label{lem:cr3}
  Let $t$ be neutral of type $A$. If every one-step reduct
  of $t$ belongs to $\SN(A)$, then $t \in \SN(A)$.
\end{lemma}
\begin{proof}
  Since every one-step reduct of $t$ is in $\SN(A) \subseteq \SN$, every
  reduction sequence from $t$ is finite, so $t \in \SN$.
  For ground and measurement types $A \in \{n, (m,n)\}$, membership in $\SN(A)$
  is just membership in $\SN$, which we just established.
  For $A = B \multimap C$: suppose $t \rightarrow^{*} \lambda x.\,u$. Since $t$
  is neutral, the first step of this sequence goes $t \rightarrow t'$ for some
  reduct $t'$, and $t' \rightarrow^{*} \lambda x.\,u$. By hypothesis,
  $t' \in \SN(B \multimap C)$, so for any $v \in \SN(B)$ we have $u[v/x] \in \SN(C)$.
  Hence $t \in \SN(B \multimap C)$.
  \qed
\end{proof}

\paragraph{Per-construct closure lemmas.}

\begin{lemma}[Closed under sums in $\SN(A)$]
  \label{lem:sn-sum-typed}
  If $t_1, \ldots, t_k \in \SN(A)$ then $\sum_i p_i\, t_i \in \SN(A)$.
\end{lemma}
\begin{proof}
  By Lemma~\ref{lem:sn-sum}, $\sum_i p_i\, t_i \in \SN$. It remains to check
  the additional condition for function types. Suppose $A = B \multimap C$ and
  $\sum_i p_i\, t_i \rightarrow^{*} \lambda x.\,u$. By the projection rules, some $t_j$
  satisfies $t_j \rightarrow^{*} \lambda x.\,u$ (since the only reduction that produces
  a $\lambda$ from a sum is via a summand). Since $t_j \in \SN(B \multimap C)$,
  for any $v \in \SN(B)$ we get $u[v/x] \in \SN(C)$. Hence
  $\sum_i p_i\, t_i \in \SN(B \multimap C)$.
  \qed
\end{proof}

\begin{lemma}[Closed under $\lambda$-abstraction]
  \label{lem:sn-lam}
  If for all $v \in \SN(A)$, $t[v/x] \in \SN(B)$, then $\lambda x.\,t \in \SN(A \multimap B)$.
\end{lemma}
\begin{proof}
  We check the two conditions in the definition of $\SN(A \multimap B)$.
  \emph{$\lambda x.\,t \in \SN$:} since the calculus has no body reduction rule,
  $\lambda x.\,t$ is a normal form and hence trivially in $\SN$.
  \emph{Body condition:} since there is no body reduction rule,
  $\lambda x.\,t \rightarrow^{*} \lambda x.\,u$ implies $u = t$. So the condition
  reduces to: for all $v \in \SN(A)$, $t[v/x] \in \SN(B)$, which is exactly the
  hypothesis.
  \qed
\end{proof}

\begin{lemma}[Closed under application]
  \label{lem:sn-app}
  If $t \in \SN(A \multimap B)$ and $s \in \SN(A)$, then $t\, s \in \SN(B)$.
\end{lemma}
\begin{proof}
  By induction on $\ell(t) + \ell(s)$ and Lemma~\ref{lem:cr3}. The term
  $t\,s$ is neutral. We show every one-step reduct of $t\,s$ is in $\SN(B)$,
  then conclude by CR3.
  \emph{Reduction inside $t$ or $s$:} if $t \rightarrow t'$ then $t' \in \SN(A \multimap B)$
  (since $\SN(A \multimap B)$ is downward closed) and $\ell(t') < \ell(t)$,
  so the induction hypothesis gives $t'\,s \in \SN(B)$; symmetrically if $s \rightarrow s'$.
  \emph{Both values, $t = \lambda x.\,u$:} the rule $(\lambda x.\,u)\,s \rightarrow u[s/x]$
  fires; since $t \rightarrow^{*} \lambda x.\,u$ (zero steps) and $s \in \SN(A)$,
  the body condition of $t \in \SN(A \multimap B)$ gives $u[s/x] \in \SN(B)$.
  \emph{Both values, $t = \sum_i p_i\,t_i$:} the distribution rule
  $(\sum_i p_i\,t_i)\,s \rightarrow \sum_i p_i\,(t_i\,s)$ fires; by the projection rules,
  $t \rightarrow t_i$ in the extended relation, so $\ell(t_i) < \ell(t)$ and
  $t_i \in \SN(A \multimap B)$ (by downward closure). The induction hypothesis
  gives each $t_i\,s \in \SN(B)$, and Lemma~\ref{lem:sn-sum-typed} then gives
  $\sum_i p_i\,(t_i\,s) \in \SN(B)$.
  In every case the one-step reduct is in $\SN(B)$, so CR3 gives $t\,s \in \SN(B)$.
  \qed
\end{proof}

\begin{lemma}[Closed under unitary operation and measurement]
  \label{lem:sn-unitary-meas}
  If $t \in \SN(n)$ then $U^m t \in \SN(n)$ and $\pi^m t \in \SN((m,n))$.
\end{lemma}
\begin{proof}
  Both $U^m t$ and $\pi^m t$ are neutral. Any reduction either applies inside
  $t$ (giving a reduct in $\SN$ by the induction hypothesis on $\ell(t)$) or
  fires the outermost rule: $U^m \rho^n \rightarrow {\rho'}^n \in \SN(n)$; and
  $\pi^m t$ has no outermost rule at all (it only participates in a reduction
  as the argument of $\mathsf{letcase}^\circ$, not by itself), so the first
  case always applies. By Lemma~\ref{lem:cr3}, both terms are in the
  appropriate $\SN$ set.
  \qed
\end{proof}

\begin{lemma}[Closed under tensor]
  \label{lem:sn-tensor}
  If $t \in \SN(m)$ and $s \in \SN(n)$, then $t \otimes s \in \SN(m+n)$.
\end{lemma}
\begin{proof}
  The term $t \otimes s$ is neutral. Reductions apply inside $t$ or $s$ (both
  stay in $\SN$ by induction), or fire the collapse rule
  $\rho_1^m \otimes \rho_2^n \rightarrow \rho^{m+n}$ (a density matrix, hence in
  $\SN(m+n)$). Lemma~\ref{lem:cr3} gives the result.
  \qed
\end{proof}

\begin{lemma}[Closed under $\mathsf{letcase}^\circ$]
  \label{lem:sn-letcase}
  If $r \in \SN((m,n))$ and, for all density matrices $\rho^n \in \SN(n)$,
  $s_i[\rho^n/x] \in \SN(A)$ for each $i$, then
  $\qletcase{x}{r}{s_0,\ldots,s_{2^m-1}} \in \SN(A)$.
\end{lemma}
\begin{proof}
  By induction on $\ell(r)$ and Lemma~\ref{lem:cr3}. The term is neutral. If
  a reduction applies inside $r$, the induction hypothesis applies.
  If $r$ is a closed value of type $(m,n)$, there are two sub-cases
  ($\lambda$-abstractions cannot have measurement types):
  \begin{itemize}
  \item $r = \pi^m \rho^n$: the main $\mathsf{letcase}^\circ$ rule fires,
    giving $\sum_i p_i\, s_i[\rho_i^n / x]$.  Each $\rho_i^n \in \SN(n)$; by
    hypothesis each $s_i[\rho_i^n/x] \in \SN(A)$; by
    Lemma~\ref{lem:sn-sum-typed} the sum is in $\SN(A)$.
  \item $r = \sum_j q_j (\pi^m \rho^n_j)$: the distributing rule fires, giving
    $\sum_j q_j\,\qletcase{x}{\pi^m\rho^n_j}{s_0,\ldots,s_{2^m-1}}$.
    Each $\pi^m\rho^n_j$ has smaller $\ell(r)$, so by the previous sub-case
    each summand is in $\SN(A)$; by
    Lemma~\ref{lem:sn-sum-typed} the sum is in $\SN(A)$.
  \end{itemize}
  Lemma~\ref{lem:cr3} concludes.
  \qed
\end{proof}

\begin{lemma}[Closed under $\mathsf{let}$]
  \label{lem:sn-let}
  If $t \in \SN(n)$ and, for all density matrices
  $\rho_k \in \SN(1)$ ($k = 1,\ldots,n$),
  $s[\rho_1/x_1,\ldots,\rho_n/x_n] \in \SN(A)$,
  then $\letxn{x}{t}{s} \in \SN(A)$.
\end{lemma}
\begin{proof}
  By induction on $\ell(t)$ and Lemma~\ref{lem:cr3}. The term is neutral. If a
  reduction applies inside $t$, the induction hypothesis applies (the reduct
  $t'$ still satisfies $t' \in \SN(n)$ as a reduct of a strongly normalising
  term). If $t$ is a value of ground type $n$, it must be a density matrix
  $\rho^n$ (no $\lambda$- or $\pi$-value has type $n$). The main $\mathsf{let}$
  rule fires:
  \[
    \letxn{x}{\rho^n}{s}
    \;\rightarrow\;
    \sum_{i,\vec{l}} p_{i\vec{l}}\,
    s[\rho_{i_1}^{l_1}/x_1, \ldots, \rho_{i_n}^{l_n}/x_n].
  \]
  Each $\rho_{i_k}^{l_k}$ is a density matrix, hence in $\SN(1)$. By the
  hypothesis, each substituted body $s[\rho_{i_1}^{l_1}/x_1, \ldots,
  \rho_{i_n}^{l_n}/x_n]$ belongs to $\SN(A)$. By
  Lemma~\ref{lem:sn-sum-typed} the entire sum belongs to $\SN(A)$.
  Lemma~\ref{lem:cr3} concludes.
  \qed
\end{proof}

\paragraph{Fundamental lemma.}
A \emph{substitution} $\theta$ is \emph{valid for} $\Gamma$, written
$\theta \vDash \Gamma$, if $\theta(x) \in \SN(A)$ for every $x : A \in \Gamma$.

\begin{lemma}[Fundamental lemma]
  \label{lem:fund-lemma}
  If $\Gamma \vdash t : A$ and $\theta \vDash \Gamma$, then $\theta(t) \in \SN(A)$.
\end{lemma}
\begin{proof}
  By induction on the typing derivation of $\Gamma \vdash t : A$.
  \begin{itemize}
    \item \textbf{Variable} $x : A \in \Gamma$. Then $\theta(x) \in \SN(A)$ by
      hypothesis.
    \item \textbf{Density matrix} $\rho^n$. Then $\theta(\rho^n) = \rho^n$,
      which is a normal form, hence in $\SN = \SN(n)$.
    \item \textbf{$\lambda$-abstraction} $\lambda x.\,u$ of type $A \multimap B$.
      Let $v \in \SN(A)$. Define $\theta' = (\theta, x \mapsto v)$; then
      $\theta' \vDash \Gamma, x : A$. By the induction hypothesis on $u$,
      $\theta'(u) = (\theta(u))[v/x] \in \SN(B)$. By
      Lemma~\ref{lem:sn-lam}, $\theta(\lambda x.\,u) = \lambda x.\,\theta(u)
      \in \SN(A \multimap B)$.
    \item \textbf{Application} $t_1\,t_2$. By the induction hypotheses,
      $\theta(t_1) \in \SN(A \multimap B)$ and $\theta(t_2) \in \SN(A)$.
      Lemma~\ref{lem:sn-app} gives $\theta(t_1\,t_2) \in \SN(B)$.
    \item \textbf{Unitary} $U^m t$. By induction, $\theta(t) \in \SN(n)$.
      Lemma~\ref{lem:sn-unitary-meas} gives $U^m\,\theta(t) \in \SN(n)$.
    \item \textbf{Measurement} $\pi^m t$. Analogous, giving
      $\pi^m\,\theta(t) \in \SN((m,n))$.
    \item \textbf{Tensor} $t_1 \otimes t_2$. By induction,
      $\theta(t_1) \in \SN(m)$ and $\theta(t_2) \in \SN(n)$.
      Lemma~\ref{lem:sn-tensor} gives $\theta(t_1) \otimes \theta(t_2) \in \SN(m+n)$.
    \item \textbf{Probabilistic sum} $\sum_i p_i\,u_i$. By induction, each
      $\theta(u_i) \in \SN(A)$. Lemma~\ref{lem:sn-sum-typed} gives
      $\theta(\sum_i p_i\,u_i) = \sum_i p_i\,\theta(u_i) \in \SN(A)$.
    \item \textbf{Letcase} $\qletcase{x}{r}{s_0,\ldots,s_{2^m-1}}$.
      By induction, $\theta(r) \in \SN((m,n))$. For any density matrix
      $\rho^n \in \SN(n)$, define $\theta'' = (\theta, x \mapsto \rho^n)$;
      then $\theta'' \vDash \Gamma, x : n$, and the induction hypothesis gives
      each $\theta(s_i)[\rho^n/x] = \theta''(s_i) \in \SN(A)$.
      Lemma~\ref{lem:sn-letcase} gives the result.
    \item \textbf{Let} $\letxn{x}{u}{s}$ with $\Gamma \vdash u : n$ and
      $\Delta, x_1:1, \ldots, x_n:1 \vdash s : A$. By induction,
      $\theta(u) \in \SN(n)$. For any density matrices $\rho_k \in \SN(1)$
      ($k = 1,\ldots,n$), define $\theta'' = (\theta|_\Delta,
      x_1 \mapsto \rho_1, \ldots, x_n \mapsto \rho_n)$; then
      $\theta'' \vDash \Delta, x_1:1, \ldots, x_n:1$, and by the induction
      hypothesis $\theta''(s) = \theta(s)[\rho_1/x_1, \ldots, \rho_n/x_n]
      \in \SN(A)$. Lemma~\ref{lem:sn-let} gives
      $\letxn{x}{\theta(u)}{\theta(s)} \in \SN(A)$.
      \qed
  \end{itemize}
\end{proof}

Taking $\theta$ to be the empty substitution (valid for any closed term), the
Fundamental Lemma~\ref{lem:fund-lemma} gives: for every closed $\vdash t : A$,
we have $t \in \SN(A) \subseteq \SN$. Hence every closed well-typed term is
strongly normalising under the extended relation, and therefore also under the
original relation.

\subsection{Semantic Substitution Lemma}

\begin{lemma}[Semantic Substitution]
  \label{lem:sem-substitution}
  Let $\Gamma, x:A \vdash t : B$, $\Delta \vdash s : A$, and
  $\theta \vDash \Gamma, \Delta$. Then
  $\fsem{t[s/x]} = \fsem[\theta, x \mapsto \fsem{s}]{t}$.
\end{lemma}

\begin{proof}
By induction on $t$.
\begin{itemize}
\item Let $t = x$. Then $t[s/x] = s$, and
  $\fsem[\theta,x\mapsto\fsem{s}]{x} = \fsem{s}$, so both sides equal $\fsem{s}$.
\item Let $t = y$ with $y \neq x$. Then $t[s/x] = y$, and
  $\fsem[\theta,x\mapsto\fsem{s}]{y} = \theta(y) = \fsem{y}$.
\item Let $t = \lambda y.u$. Then $t[s/x] = \lambda y.(u[s/x])$
  (assuming $y \neq x$ and $y \notin FV(s)$ by renaming). By definition of
  the interpretation and the induction hypothesis on $u$:
  \begin{align*}
    \fsem{\lambda y.(u[s/x])}
    &= \rho \mapsto \fsem[\theta,\, y \mapsto \rho]{u[s/x]}
     = \rho \mapsto \fsem[\theta,\, x \mapsto \fsem{s},\, y \mapsto \rho]{u}
     = \fsem[\theta,x\mapsto\fsem{s}]{\lambda y.u}.
  \end{align*}
\item Let $t = t_1 t_2$, with $\Gamma, x:A = \Gamma_1, \Gamma_2$
  and $x:A \in \Gamma_j$. By the induction hypothesis applied to
  the subterm containing $x$:
  \[
    \fsem{(t_1 t_2)[s/x]}
    = \fsem{t_1[s/x]}(\fsem{t_2[s/x]})
    = \fsem[\theta,x\mapsto\fsem{s}]{t_1}\bigl(\fsem[\theta,x\mapsto\fsem{s}]{t_2}\bigr)
    = \fsem[\theta,x\mapsto\fsem{s}]{t_1 t_2}.
  \]
\item Let $t = \rho^n$. The variable $x$ does not appear, so
  $t[s/x] = \rho^n$ and $\fsem[\theta,x\mapsto\fsem{s}]{\rho^n} = \rho^n = \fsem{\rho^n}$.
\item Let $t = U^m u$. Then $(U^m u)[s/x] = U^m(u[s/x])$. By the induction hypothesis:
  \[
    \fsem{U^m(u[s/x])}
    = \ext{U^m}\fsem{u[s/x]}\ext{U^m}^\dagger
    = \ext{U^m}\fsem[\theta,x\mapsto\fsem{s}]{u}\,\ext{U^m}^\dagger
    = \fsem[\theta,x\mapsto\fsem{s}]{U^m u}.
  \]
\item Let $t = \pi^m u$. Then $(\pi^m u)[s/x] = \pi^m(u[s/x])$.
  By the induction hypothesis:
  \[
    \fsem{\pi^m(u[s/x])}
    = \sum_i \ext{\pi_i}\fsem{u[s/x]}\ext{\pi_i}^\dagger
    = \sum_i \ext{\pi_i}\fsem[\theta,x\mapsto\fsem{s}]{u}\,\ext{\pi_i}^\dagger
    = \fsem[\theta,x\mapsto\fsem{s}]{\pi^m u}.
  \]
\item Let $t = t_1 \otimes t_2$, with $\Gamma, x:A = \Gamma_1, \Gamma_2$
  and $x:A \in \Gamma_j$. By the induction hypothesis on the relevant subterm:
  \begin{align*}
    \fsem{(t_1 \otimes t_2)[s/x]}
    &= \fsem{t_1[s/x]} \otimes \fsem{t_2[s/x]} \\
    &= \fsem[\theta,x\mapsto\fsem{s}]{t_1} \otimes \fsem[\theta,x\mapsto\fsem{s}]{t_2}
    = \fsem[\theta,x\mapsto\fsem{s}]{t_1 \otimes t_2}.
  \end{align*}
\item Let $t = \sum_j q_j t_j$. Then $t[s/x] = \sum_j q_j (t_j[s/x])$.
  By the induction hypothesis and linearity of the interpretation:
  \[
    \fsem{\sum_j q_j (t_j[s/x])}
    = \sum_j q_j \fsem{t_j[s/x]}
    = \sum_j q_j \fsem[\theta,x\mapsto\fsem{s}]{t_j}
    = \fsem[\theta,x\mapsto\fsem{s}]{\sum_j q_j t_j}.
  \]
\item Let $t = \qletcase{y}{r}{t_0,\ldots,t_{2^m-1}}$,
  with $\Gamma' \vdash r : (m,n)$ and $y:n \vdash t_i : B$.
  Since the $t_i$ are typed in $y:n$ alone, $x \notin FV(t_i)$,
  so $t[s/x] = \qletcase{y}{r[s/x]}{t_0,\ldots,t_{2^m-1}}$.
  By the induction hypothesis on $r$,
  $\fsem{r[s/x]} = \fsem[\theta,x\mapsto\fsem{s}]{r}$,
  so the post-measurement states $\rho_i$ and probabilities $p_i$ are the
  same on both sides:
  $\fsem{t[s/x]} = \sum_i p_i\, \fsem[\theta,y\mapsto\rho_i]{t_i}
  = \fsem[\theta,x\mapsto\fsem{s}]{t}$.
\item Let $t = \letxn{y}{u}{v}$, with $\Gamma, x:A = \Gamma_1, \Gamma_2$,
  $\Gamma_1 \vdash u : n$ and $\Gamma_2, y_1:1,\ldots,y_n:1 \vdash v : B$
  (with $y_k \neq x$ by $\alpha$-renaming).
  \begin{itemize}
  \item If $x:A \in \Gamma_1$, then $t[s/x] = \letxn{y}{u[s/x]}{v}$,
    since $x \notin FV(v)$.
    By the induction hypothesis on $u$,
    $\fsem{u[s/x]} = \fsem[\theta,x\mapsto\fsem{s}]{u}$,
    so the combined decomposition coefficients $\alpha_i$ are the same on both
    sides:
    \begin{align*}
      \fsem{t[s/x]}
      &= \sum_{i,\vec{l}} p_{i\vec{l}}\,
         \fsem[\theta,\, y_1\mapsto\rho_{i_1}^{l_1},\,\ldots,\, y_n\mapsto\rho_{i_n}^{l_n}]{v}
       = \fsem[\theta,x\mapsto\fsem{s}]{t}.
    \end{align*}
  \item If $x:A \in \Gamma_2$, then $t[s/x] = \letxn{y}{u}{v[s/x]}$,
    since $x \notin FV(u)$.
    The combined decomposition coefficients of $\fsem{u}$ are unchanged.
    By the induction hypothesis on $v$:
    \begin{align*}
      \fsem{t[s/x]}
      &= \sum_{i,\vec{l}} p_{i\vec{l}}\,
         \fsem[\theta,\, x\mapsto\fsem{s},\, y_1\mapsto\rho_{i_1}^{l_1},\,\ldots,\,
               y_n\mapsto\rho_{i_n}^{l_n}]{v}
       = \fsem[\theta,x\mapsto\fsem{s}]{t}.
       &\tag*{\qed}
    \end{align*}
  \end{itemize}
\end{itemize}
\end{proof}

\subsection{Soundness (Theorem~\ref{thm:Soundness})}

\begin{proof}
By induction on the derivation of $t \rightarrow r$.
\begin{itemize}
\item Let $t = (\lambda x.t')s$, $r = t'[s/x]$.
  \[
    \fsem{(\lambda x.t')s}
    = \fsem{\lambda x.t'}(\fsem{s})
    = \fsem[\theta,x\mapsto\fsem{s}]{t'}
    = \fsem{t'[s/x]},
  \]
  where the last step uses Lemma~\ref{lem:sem-substitution}.
\item Let $t = U^m\rho^n$, $r = {\rho'}^n$ with
  ${\rho'}^n = \ext{U^m}\rho^n\ext{U^m}^\dagger$.
  $\fsem{U^m\rho^n} = \ext{U^m}\rho^n\ext{U^m}^\dagger = {\rho'}^n = \fsem{{\rho'}^n}$.
\item Let $t = \rho_1^m \otimes \rho_2^n$, $r = \rho$ with
  $\rho = \rho_1^m \otimes \rho_2^n$.
  $\fsem{\rho_1^m \otimes \rho_2^n}
    = \fsem{\rho_1^m} \otimes \fsem{\rho_2^n}
    = \rho_1^m \otimes \rho_2^n = \rho = \fsem{\rho}$.
\item Let $t = \sum_i p_i \rho_i$, $r = \rho'$ with $\rho' = \sum_i p_i\rho_i$.
  $\fsem{\sum_i p_i \rho_i} = \sum_i p_i\rho_i = \rho' = \fsem{\rho'}$.
\item Let $t = \sum_i p_i u$, $r = u$.
  $\fsem{\sum_i p_i u} = \sum_i p_i\fsem{u} = \bigl(\textstyle\sum_i p_i\bigr)\fsem{u}
    = \fsem{u}$.
\item Let $t = (\sum_i p_i t_i)r'$, $r = \sum_i p_i(t_i r')$.
  \begin{align*}
    \fsem{(\sum_i p_i t_i)r'}
    &= \Bigl(\sum_i p_i\fsem{t_i}\Bigr)(\fsem{r'})
    = \sum_i p_i\bigl(\fsem{t_i}(\fsem{r'})\bigr) \\
    &= \sum_i p_i\fsem{t_i r'}
    = \fsem{\sum_i p_i(t_i r')}.
  \end{align*}
\item Let $t = \qletcase{x}{\pi^m\rho^n}{t_0,\ldots,t_{2^m-1}}$,
  $r = \sum_i p_i\,t_i[\rho_i^n/x]$ with
  $p_i = \tr(\ext{\pi_i}^\dagger\ext{\pi_i}\rho^n)$ and
  $\rho_i^n = \ext{\pi_i}\rho^n\ext{\pi_i}^\dagger/p_i$.
  Using linearity and Lemma~\ref{lem:sem-substitution}:
  \[
    \fsem{r}
     = \sum_i p_i\,\fsem{t_i[\rho_i^n/x]}
     = \sum_i p_i\,\fsem[\theta,x\mapsto\rho_i^n]{t_i}
     = \fsem{t}.
  \]
\item Let $t = \qletcase{x}{\sum_j q_j w_j}{t_0,\ldots,t_{2^m-1}}$ and
  $r = \sum_j q_j\,\qletcase{x}{w_j}{t_0,\ldots,t_{2^m-1}}$.
  Write $F_i(\rho) \coloneqq \fsem[\theta,x\mapsto\rho]{t_i}$; this is linear
  in $\rho$ because the type system is affine ($x$ appears at most once in
  $t_i$).  Let $\rho = \fsem{\sum_j q_j w_j} = \sum_j q_j\fsem{w_j}$, and
  define $p_i = \tr(\ext{\pi_i}^\dagger\ext{\pi_i}\rho)$ and
  $\rho_i = \ext{\pi_i}\rho\,\ext{\pi_i}^\dagger/p_i$ as in the semantics.
  Note that $p_i\rho_i = \ext{\pi_i}\rho\,\ext{\pi_i}^\dagger$, so by linearity
  of $F_i$, $p_i\cdot F_i(\rho_i) = F_i(\ext{\pi_i}\rho\,\ext{\pi_i}^\dagger)$.
  Therefore, substituting $\rho = \sum_j q_j\fsem{w_j}$ and using linearity:
  \begin{align*}
    \fsem{t}
    &= \sum_i p_i\cdot F_i(\rho_i)
    = \sum_i F_i\!\bigl(\ext{\pi_i}\rho\,\ext{\pi_i}^\dagger\bigr) \\
    &= \sum_j q_j \sum_i F_i\!\bigl(\ext{\pi_i}\fsem{w_j}\ext{\pi_i}^\dagger\bigr).
  \end{align*}
  For each fixed $j$, letting $p_i^j = \tr(\ext{\pi_i}^\dagger\ext{\pi_i}\fsem{w_j})$
  and $\rho_i^j = \ext{\pi_i}\fsem{w_j}\ext{\pi_i}^\dagger/p_i^j$, linearity of $F_i$ gives
  \[\sum_i F_i(\ext{\pi_i}\fsem{w_j}\ext{\pi_i}^\dagger)
    = \sum_i p_i^j F_i(\rho_i^j)
    = \fsem{\qletcase{x}{w_j}{t_0,\ldots,t_{2^m-1}}}.\]
  Hence $\fsem{t} = \sum_j q_j\,\fsem{\qletcase{x}{w_j}{t_0,\ldots,t_{2^m-1}}} = \fsem{r}$.
\item Let $t = \letxn{x}{\rho^n}{s}$,
  $r = \sum_{i,\vec{l}} p_{i\vec{l}}\,s[\rho_{i_1}^{l_1}/x_1,
  \ldots,\rho_{i_n}^{l_n}/x_n]$,
  where $p_{i\vec{l}} = \alpha_i(\rho^n)\prod_k\lambda_{i_k}^{l_k}$.
  Unfolding the interpretation of $\mathsf{let}$ and applying
  Lemma~\ref{lem:sem-substitution} $n$ times:
  \begin{align*}
    \fsem{t}
    &= \sum_{i,\vec{l}} p_{i\vec{l}}\,
       \fsem[\theta,\,x_1\mapsto\rho_{i_1}^{l_1},\,\ldots,\,x_n\mapsto\rho_{i_n}^{l_n}]{s}
     = \sum_{i,\vec{l}} p_{i\vec{l}}\,
       \fsem{s[\rho_{i_1}^{l_1}/x_1,\ldots,\rho_{i_n}^{l_n}/x_n]}
     = \fsem{r}.
  \end{align*}
\item Contextual cases. Let $s \rightarrow s'$.
  \begin{itemize}
  \item Let $t = s\,u$, $r = s'\,u$.
    By the induction hypothesis $\fsem{s} = \fsem{s'}$, so
    $\fsem{s\,u} = \fsem{s}(\fsem{u}) = \fsem{s'}(\fsem{u}) = \fsem{s'\,u}$.
  \item Let $t = u\,s$, $r = u\,s'$.
    By the induction hypothesis $\fsem{s} = \fsem{s'}$, so
    $\fsem{u\,s} = \fsem{u}(\fsem{s}) = \fsem{u}(\fsem{s'}) = \fsem{u\,s'}$.
  \item Let $t = U^m s$, $r = U^m s'$.
    $\fsem{U^m s} = \ext{U^m}\fsem{s}\ext{U^m}^\dagger
    = \ext{U^m}\fsem{s'}\ext{U^m}^\dagger = \fsem{U^m s'}$.
  \item Let $t = \pi^m s$, $r = \pi^m s'$.
    $\fsem{\pi^m s} = \sum_i\ext{\pi_i}\fsem{s}\ext{\pi_i}^\dagger
    = \sum_i\ext{\pi_i}\fsem{s'}\ext{\pi_i}^\dagger = \fsem{\pi^m s'}$.
  \item Let $t = s \otimes u$, $r = s' \otimes u$.
    $\fsem{s \otimes u} = \fsem{s}\otimes\fsem{u}
    = \fsem{s'}\otimes\fsem{u} = \fsem{s'\otimes u}$. The case
    $t = u \otimes s$ is analogous.
  \item Let $t = \sum_j p_j u_j$, $r = \sum_j p_j v_j$ where $u_k \rightarrow v_k$
    and $u_j = v_j$ for $j \neq k$.
    By the induction hypothesis $\fsem{u_k} = \fsem{v_k}$, so
    $\fsem{\sum_j p_j u_j} = \sum_j p_j\fsem{u_j}
    = \sum_j p_j\fsem{v_j} = \fsem{\sum_j p_j v_j}$.
  \item Let $t = \qletcase{x}{s}{u_0,\ldots,u_{2^m-1}}$,
    $r = \qletcase{x}{s'}{u_0,\ldots,u_{2^m-1}}$.
    By the induction hypothesis $\fsem{s} = \fsem{s'}$, so the
    post-measurement states and probabilities are the same, giving $\fsem{t} = \fsem{r}$.
  \item Let $t = \letxn{x}{s}{u}$, $r = \letxn{x}{s'}{u}$.
    By the induction hypothesis $\fsem{s} = \fsem{s'}$, so
    $\alpha_i(\fsem{s}) = \alpha_i(\fsem{s'})$ for all $i$, and the two
    interpretations agree term by term.
  \qed
  \end{itemize}
\end{itemize}
\end{proof}

\subsection{Compositionality and Adequacy}

\begin{lemma}[Compositionality]
  \label{lem:compositionality}
  Let $C$ be a context with hole type $A$ and output type $B$, and let $\theta$
  be any valuation of the free variables of $C$ (those not contributed by the
  hole). For any closed terms $\vdash t : A$ and $\vdash r : A$,
  $\fsem{t} = \fsem{r}$ implies $\fsem[\theta]{C[t]} = \fsem[\theta]{C[r]}$.
\end{lemma}

\begin{proof}
  We prove the lemma by structural induction on the context $C$. In each case
  we write $t$ and $r$ for the two closed terms with $\fsem{t} = \fsem{r}$, and
  $\theta$ for any fixed valuation of the free variables of $C$.
  \begin{itemize}
  \item Let $C = [\cdot]$.
    $C[t] = t$ and $C[r] = r$, so
    $\fsem[\theta]{C[t]} = \fsem{t} = \fsem{r} = \fsem[\theta]{C[r]}$
    directly from the hypothesis.
  \item Let $C = \lambda x.C'$. For any $v$ in the domain of $x$, by the
    induction hypothesis applied to $C'$ with valuation $\theta' = (\theta, x
    \mapsto v)$, $\fsem[\theta']{C'[t]} = \fsem[\theta']{C'[r]}$. Hence the two
    functions $v \mapsto \fsem[\theta']{C'[t]}$ and $v \mapsto
    \fsem[\theta']{C'[r]}$ agree pointwise, giving $\fsem[\theta]{\lambda
    x.C'[t]} = \fsem[\theta]{\lambda x.C'[r]}$.
  \item Let $C = C'\,s$.
    $\fsem[\theta]{C'[t]\,s} = \fsem[\theta]{C'[t]}(\fsem[\theta]{s})$, and by
    the induction hypothesis $\fsem[\theta]{C'[t]} = \fsem[\theta]{C'[r]}$, so
    the two are equal. The case $C = s\,C'$ is symmetric.
  \item Let $C = U^m C'$.
    $\fsem[\theta]{U^m C'[t]} = \ext{U^m}\,\fsem[\theta]{C'[t]}\,\ext{U^m}^\dagger
    = \ext{U^m}\,\fsem[\theta]{C'[r]}\,\ext{U^m}^\dagger = \fsem[\theta]{U^m C'[r]}$,
    using the induction hypothesis.
  \item Let $C = \pi^m C'$.
    $\fsem[\theta]{\pi^m C'[t]} = \sum_{i} \ext{\pi_i}\,\fsem[\theta]{C'[t]}\,\ext{\pi_i}^\dagger
    = \sum_{i} \ext{\pi_i}\,\fsem[\theta]{C'[r]}\,\ext{\pi_i}^\dagger = \fsem[\theta]{\pi^m C'[r]}$.
  \item Let $C = C' \otimes s$.
    $\fsem[\theta]{C'[t] \otimes s} = \fsem[\theta]{C'[t]} \otimes \fsem[\theta]{s}
    = \fsem[\theta]{C'[r]} \otimes \fsem[\theta]{s} = \fsem[\theta]{C'[r] \otimes s}$.
    The case $C = s \otimes C'$ is symmetric.
  \item Let $C = \sum_i p_i s_i$ with the hole in $s_j$. Then
    $\fsem[\theta]{\sum_i p_i s_i}
      = \sum_{i \neq j} p_i \fsem[\theta]{s_i} + p_j \fsem[\theta]{C'[t]}$,
    and by the induction hypothesis $\fsem[\theta]{C'[t]} = \fsem[\theta]{C'[r]}$,
    so the sum is unchanged.
  \item Let $C = \qletcase{x}{C'}{t_0,\ldots,t_{2^m-1}}$.
    Let $\sigma = \fsem[\theta]{C'[t]}$ and $\sigma' = \fsem[\theta]{C'[r]}$.
    By the induction hypothesis $\sigma = \sigma'$. The probabilities
    $p_i = \tr(\ext{\pi_i}^\dagger \ext{\pi_i}\,\sigma)$ and post-measurement
    states $\sigma_i = \ext{\pi_i}\,\sigma\,\ext{\pi_i}^\dagger / p_i$ depend
    only on $\sigma$, so they are the same for $C'[t]$ and $C'[r]$. Hence
    $\fsem[\theta]{\qletcase{x}{C'[t]}{t_0,\ldots,t_{2^m-1}}}
      = \sum_i p_i \fsem[\theta, x\mapsto \sigma_i]{t_i}
      = \fsem[\theta]{\qletcase{x}{C'[r]}{t_0,\ldots,t_{2^m-1}}}$.
  \item Let $C = \qletcase{x}{s}{t_0,\ldots,C',\ldots,t_{2^m-1}}$.
    The probabilities $p_i$ and post-measurement states $\sigma_i$ of $s$ are
    independent of the hole. For the $j$-th branch, by the induction hypothesis
    (with valuation $\theta, x \mapsto \sigma_j$),
    $\fsem[\theta, x\mapsto\sigma_j]{C'[t]} = \fsem[\theta, x\mapsto\sigma_j]{C'[r]}$.
    All other branches are unchanged, so the total expression is unchanged.
  \item Let $C = \letxn{x}{C'}{s}$.
    The combined decomposition of $\letxn{x}{C'[t]}{s}$ is computed from
    $\fsem[\theta]{C'[t]}$. By the induction hypothesis,
    $\fsem[\theta]{C'[t]} = \fsem[\theta]{C'[r]}$, so the Pauli coefficients
    $\alpha_i$ and hence the weights $p_{i\vec{l}}$ and eigenprojectors
    $\rho_{i_k}^{l_k}$ are identical for both sides. Therefore
    $\fsem[\theta]{\letxn{x}{C'[t]}{s}}
      = \sum_{i,\vec{l}} p_{i\vec{l}}\, \fsem[\theta, x_1\mapsto\rho_{i_1}^{l_1}, \ldots, x_n\mapsto\rho_{i_n}^{l_n}]{s}
      = \fsem[\theta]{\letxn{x}{C'[r]}{s}}$.
  \item Let $C = \letxn{x}{s}{C'}$.
    The combined decomposition of $s$ is fixed. For each index pair
    $(i, \vec{l})$, let $\theta_{i\vec{l}} = \theta, x_1 \mapsto
    \rho_{i_1}^{l_1}, \ldots, x_n \mapsto \rho_{i_n}^{l_n}$. By the induction
    hypothesis applied to $C'$ with valuation $\theta_{i\vec{l}}$,
    $\fsem[\theta_{i\vec{l}}]{C'[t]} = \fsem[\theta_{i\vec{l}}]{C'[r]}$.
    Multiplying by $p_{i\vec{l}}$ and summing over all $(i, \vec{l})$ gives
    $\fsem[\theta]{\letxn{x}{s}{C'[t]}} = \fsem[\theta]{\letxn{x}{s}{C'[r]}}$.
  \end{itemize}
  In every case $\fsem[\theta]{C[t]} = \fsem[\theta]{C[r]}$, completing the
  induction.
  \qed
\end{proof}

\begin{lemma}[Normalisation]
  \label{lem:strnorm}
  If $\vdash t : n$ is a closed well-typed term of type $n$, then there exists a
  density matrix $\rho^n$ such that $t \rightarrow^{*} \rho^n$.
\end{lemma}
\begin{proof}
  By Theorem~\ref{thm:strnorm} (Strong Normalisation), $t$ is strongly
  normalising. By Theorem~\ref{thm:Progress} (Progress) and
  Theorem~\ref{thm:SR} (Subject Reduction), every closed well-typed term either
  is a value or reduces, with type preserved. Since $t$ has type $n$, its normal
  form must be a value of type $n$, which can only be a density matrix $\rho^n$.
  \qed
\end{proof}

\begin{proof}[of Theorem~\ref{thm:Adequacy} (Adequacy)]
  The argument is assembled directly from the two lemmas and soundness.
  Let $C$ be any context with hole type $A$ and output type $n$ (so
  $\vdash C[t] : n$ and $\vdash C[r] : n$ are closed). We exhibit a density
  matrix $\rho$ such that $C[t] \rightarrow^{*} \rho$ and $C[r] \rightarrow^{*} \rho$.

  \emph{Step 1 (Compositionality).}
  By Lemma~\ref{lem:compositionality} applied with the empty valuation and the
  hypothesis $\fsem{t} = \fsem{r}$, $\fsem{C[t]} = \fsem{C[r]}$.

  \emph{Step 2 (Normalisation).}
  By Lemma~\ref{lem:strnorm}, since $\vdash C[t] : n$ and $\vdash C[r] : n$ are
  closed ground-type terms, there exist density matrices $\rho_1$ and $\rho_2$
  with $C[t] \rightarrow^{*} \rho_1$ and $C[r] \rightarrow^{*} \rho_2$.

  \emph{Step 3 (Soundness, iterated).}
  Writing the first sequence as $C[t] = s_0 \rightarrow s_1 \rightarrow \cdots \rightarrow s_k = \rho_1$,
  Theorem~\ref{thm:Soundness} (Soundness) applied at each step gives $\fsem{s_0}
  = \cdots = \fsem{s_k}$, so $\fsem{C[t]} = \fsem{\rho_1} = \rho_1$. Likewise
  $\fsem{C[r]} = \rho_2$.

  \emph{Step 4 (Identification).}
  $\rho_1 = \fsem{\rho_1} = \fsem{C[t]} = \fsem{C[r]} = \fsem{\rho_2} = \rho_2$.
  Setting $\rho \coloneqq \rho_1 = \rho_2$, we have $C[t] \rightarrow^{*} \rho$ and
  $C[r] \rightarrow^{*} \rho$. Since $C$ was an arbitrary context with output type
  $n$, Definition~\ref{def:obs-equiv} gives $t \equiv r$.
  \qed
\end{proof}

\end{document}